\makeatletter
\@ifundefined{@parse@version@dash}{%
\def\@parse@version#1{\@parse@version@0#1}
\def\@parse@version@#1/#2/#3#4#5\@nil{%
\@parse@version@dash#1-#2-#3#4\@nil}
\def\@parse@version@dash#1-#2-#3#4#5\@nil{%
\if\relax#2\relax\else#1\fi#2#3#4 }
}{}
\makeatother

\documentclass[preprint,review,12pt]{elsarticle}

\usepackage{amssymb}
\usepackage{latexsym}
\usepackage{wrapfig}
\usepackage{amsmath}
\usepackage{booktabs}
\usepackage{threeparttable}
\usepackage{xcolor}
\usepackage{amsthm}
\usepackage{comment}
\usepackage{graphicx, caption, subcaption}
\usepackage{titlesec,enumitem}
\usepackage{xcolor}
\usepackage{makecell}
\usepackage{placeins}
\usepackage{caption}
\usepackage{footnote}
\usepackage[colorlinks,linkcolor=red,anchorcolor=blue,citecolor=green]{hyperref}
\usepackage{epstopdf}

\biboptions{numbers,sort&compress}

\newcommand{\bfb}{{\bf b}}

\newcommand{\bfn}{{\bf n}}

\newcommand{\bfr}{{\bf r}}

\newcommand{\bfv}{{\bf v}}

\newcommand{\bfy}{{\bf y}}

\newcommand{\bfR}{{\bf R}}

\newcommand{\beq}{\begin{equation}}
\newcommand{\eeq}{\end{equation}}
\newcommand{\beqs}{\begin{eqnarray}}
\newcommand{\eeqs}{\end{eqnarray}}

\newtheorem{lemma}{Lemma}[section]

\newtheorem{theorem*}{Theorem}

\begin{document}

\begin{frontmatter}



\title{Inverse design of deployable origami structures that approximate a general surface}


\author[1]{Xiangxin Dang}
\author[2]{Fan Feng}
\author[3]{Paul Plucinsky}
\author[4]{Richard D. James}
\author[1]{Huiling Duan}
\author[1]{Jianxiang Wang\corref{mycorrespondingauthor}}
\cortext[mycorrespondingauthor]{Corresponding author}
\ead{jxwang@pku.edu.cn}

\address[1]{State Key Laboratory for Turbulence and Complex Systems, Department of Mechanics and Engineering Science, College of Engineering, Peking University, Beijing 100871, China}
\address[2]{Cavendish Laboratory, University of Cambridge, Cambridge CB3 0HE, UK}
\address[3]{Aerospace and Mechanical Engineering, University of Southern California, Los Angeles, California 90089, USA}
\address[4]{Aerospace Engineering and Mechanics, University of Minnesota, Minneapolis, MN 55455, USA}

\begin{abstract}
Shape-morphing finds widespread utility, from the deployment of small stents and large solar sails to actuation and propulsion in soft robotics.  Origami structures provide a template for shape-morphing, but rules for designing and folding the structures are challenging to integrate into broad and versatile design tools. Here, we develop a sequential two-stage optimization framework to approximate a general surface by  a deployable origami structure. The optimization is performed over the space of all possible rigidly and flat-foldable quadrilateral mesh origami. So, the origami structures produced by our framework come with desirable engineering properties: they can be easily manufactured on a flat reference sheet, deployed to their target state by a controlled folding motion, then to a compact folded state in applications involving storage and portability. The attainable surfaces demonstrated include those with modest but diverse curvatures and unprecedented ones with sharp ridges.  The framework provides not only a tool to design various deployable and retractable surfaces in engineering and architecture, but also a route to optimizing other properties and functionality.
\end{abstract}

\begin{keyword}
origami \sep inverse design \sep optimization \sep deployability


\end{keyword}

\end{frontmatter}


\section{Introduction}

Origami is the art of paper folding, long appreciated for its aesthetic quality \cite{lang2011origami}. Interest in the sciences and engineering has followed \cite{huffman1976curvature, miura1985method, kawasaki1991relation, hull1994mathematics, tachi2009generalization, filipov2015origami, CALLENS2018241, li2019Origami, gu2020Origami}. Origami is now seen as a tool for large and coordinated shape-morphing increasingly sought in many applications.  With the right folding patterns, one can achieve rapid deployment across scales from medical stents to reconfigurable antennas and solar sails \cite{kuribayashi2006self, zirbel2013accommodating, pellegrino2014deployable}. Origami is also useful as a mechanism for robotic motion \cite{felton2014method, kim2018printing} or as a way to assemble complex surfaces in manufacturing \cite{rogers2016origami}. Yet, despite this promise, inverse design in origami --- the process of arranging rigid panels and straight-line creases into a pattern that can be folded to achieve desired configurations in space --- is hindered by delicate and nonlinear constraints. So it is challenging to develop broad design principles for origami structures that balance considerations of practical engineering and inverse design. Approaches based on symmetry and offshoots thereof \cite{Gattas2013Miura, Sareh2015Designofisomorphic, Sareh2015Designofnon, Hu2019Design, feng2020helical, mcinerney2020hidden} are ideal for manufacturability and foldability,  but generally lack the versatility needed for inverse design.  Approaches based on space-filling algorithms and fine-scale fold operations \cite{lang1996computational, demaine2017origamizer}  are adept at handling the inverse problem, but typically yield crease patterns that are difficult to fold by actuation or mechanical control systems common to practical engineering.  And approaches based on fixing a crease pattern topology \cite{tachi2010freeform, Levi2016Programming, pratapa2019geometric, dieleman2020jigsaw, dudte2020additive, Hu2020Rigid, hayakawa2020Form}, while attempting to strike this balance,  typically only perform well for one of the following criteria or the other:
\begin{itemize}[leftmargin=*]
\item \textit{Controlled deployability.}~A framework for producing designs that exhibit coordinated shape-morphing, i.e., the ability to deploy as a mechanism by a controlled folding motion from one state to another. This functional property is needed in many applications, yet far from guaranteed in origami design. 
\item \textit{Versatility for inverse design.} A framework that is general enough to address a broad range of inverse design problems.
\end{itemize}
With this work, we seek a novel design framework for origami structures that both guarantees controlled deployability and is versatile for inverse design.

Among various types of origami structures constructed with polygonal-mesh crease patterns, rigidly and flat-foldable quadrilateral mesh origami (RFFQM) is one special class with two fundamental properties. First,  the origami can be initially designed and manufactured on a flat reference domain, deployed to its  target state, then finally to a compact folded flat state, all without any stretching or bending of the panels during the entire process. In addition, the compact folded state can be unfolded to the target state in application involving storage and portability. Second, the folding kinematics have only one  degree-of-freedom (DOF); all the folding angles vary in a coordinated manner during the folding process. This feature simplifies the design of a mechanical control system or actuation strategy.  { Given these properties, we consider RFFQM to be a promising template for the design of \textit{deployable structures}, and develop our inverse design strategy using this class of origami}.

The canonical example of RFFQM is the famed Miura-Ori. This origami often serves as a  paradigm to demonstrate the efficacy and utility of folding strategies \cite{miura1985method, schenk2013geometry, wei2013geometric, silverberg2014using, na2015programming}, yet it is just a singular example in a much larger design space.   RFFQM are now thoroughly characterized~\cite{tachi2009generalization, lang2018rigidly, feng2020designs}. Our previous work \cite{feng2020designs} employed the concept of rank-one compatibility \cite{ball1989fine,bhattacharya2003microstructure,song2013enhanced} to derive an explicit marching algorithm for the designs and deformations of all possible RFFQM.  Here, we demonstrate that this marching algorithm can be  an effective ingredient to inverse design when supplemented with a careful optimization procedure.     Importantly, to the success of this strategy, the configuration space of all RFFQM is quite broad{: While research on RFFQM has often focused on symmetric Miura-Ori like patterns \cite{Gattas2013Miura, Sareh2015Designofisomorphic, Sareh2015Designofnon, Hu2019Design}, which exhibit simple (planar/cylindrical) modes of deformation on folding,  generic RFFQM patterns can have significant spatial variations in their crease design, thus enriching their possible deformation modes. We can therefore explore the configuration space of RFFQM systematically using the marching algorithm with the goal of approximating a variety of surfaces.}

To this end, we develop a general, efficient and widely applicable inverse design framework  to achieve a targeted surface by controlled deployment of a RFFQM crease pattern.
The design process is composed of two progressive optimization steps to pursue the best approximation of the targeted surface while strictly guaranteeing the non-linear constraints induced from rigid and flat foldability. By optimizing the input parameters of the marching algorithm described above, we minimize the difference between the shape of an origami structure produced by RFFQM and the targeted surface to achieve the optimal design.
The article is arranged as follows. We first recall the marching algorithm that parameterizes all possible RFFQM. Then we present the schematic of our inverse design method.  As illustrations of the approach, diverse examples of surfaces with varying curvatures and even sharp ridges are presented.   We also highlight the versatility of our optimization framework by extending it to design more general quad-mesh origami that accurately approximates a human face.  

\begin{figure}[t!]
\includegraphics[width=0.9\linewidth]{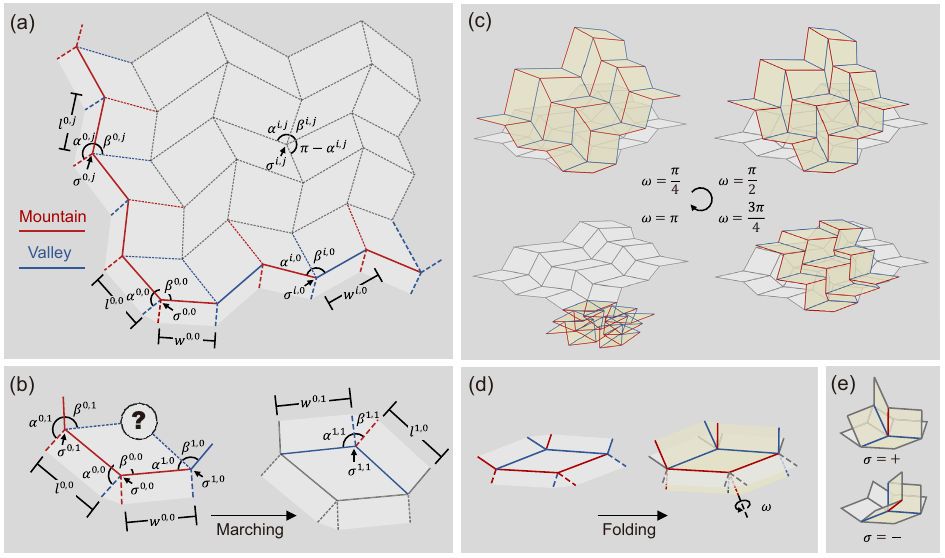}
\centering
\caption{Marching algorithm for deployable origami.  (a) The design of a RFFQM crease pattern is fully determined by input data provided at the left and bottom boundary (i.e., the red and blue solid lines and M-V assignment at each boundary vertex), which is a collection of angles, lengths, and signs encoding the M-V assignment on the ``L"-shaped outline. (b) The algorithm is initialized at the lower left corner.  Since input data at three of four vertices is provided, the fourth vertex is determined by foldability.  The overall crease pattern is then obtained by marching and making repeated use of this basic fact. (c) The crease pattern emerging from this algorithm is RFFQM.  The pattern's deformations are characterized by  a single DOF --- a continuous motion in $\omega$ from flat ($\omega=0$) to folded-flat ($\omega=\pi$). { (d) The folding parameter $\omega$ is defined by the folding angle of the crease below the vertex at the lower left corner, i.e., the red dashed crease in this figure. (e)} The M-V assignment for the motion is indicated by the choice of $+$ or $-$ at each vertex.}
\label{fig:forward}
\end{figure}

\section{Inverse design framework}

\subsection{Marching algorithm for deployable origami by RFFQM}

To begin, we recall the characterization of the designs and deformations of all RFFQM via the marching algorithm derived in \cite{feng2020designs}. This algorithm will be a key ingredient to our inverse design framework. 

A quadrilateral mesh crease pattern is comprised of quad panels ($M$ columns and $N$ rows) arranged in a plane and connected along creases. RFFQM is a special class of quad-mesh origami patterns with the desirable properties for deployment that restrict their design. Specifically and as illustrated in Fig.~\ref{fig:forward}(a), a RFFQM crease pattern with $M\times N$ panels is characterized by two sector angles at each vertex: $0 < \alpha^{i,j}, \beta^{i,j} < \pi$,  for $i=0,1,\ldots,M$, $j=0,1,\ldots,N$.  The other two sector angles at each vertex are constrained  so that the sum of all four angles is $2\pi$ (developability) and  the sum of opposite angles is $\pi$ (flat-foldability/Kawasaki's condition).  These conditions are necessary for RFFQM but far from sufficient.

Here, we address sufficiency using a marching algorithm that is initialized by the input data indicated schematically by the red and blue line segments in Fig.~\ref{fig:forward}(a).  This data is  a collection of all of the angles $0 < \alpha^{i,0}, \beta^{i,0},\alpha^{0,j}, \beta^{0,j}< \pi$, lengths $w^{i,0}, l^{0,j} >0$ and signs $\sigma^{i,0}, \sigma^{0,j} = +$ or $-$ that parameterize the left and bottom ``L''-shaped boundary creases on the $M \times N$ pattern. Note, the signs encode valid {mountain-valley (M-V)} assignments ({ Fig.~\ref{fig:forward}(e)}) at each vertex on the ``L'', and the exact formula relating signs to the M-V assignments is provided in \ref{sect:ap-marching}. From hereon, we represent this data compactly through arrays $\boldsymbol{\alpha}_0, \boldsymbol{l}_0$ and $\boldsymbol{\sigma}_0$ that list all such boundary angles, lengths and signs, respectively.

We proved in \cite{feng2020designs} that it is possible to march algorithmically and discover that:
\begin{theorem*}
For any input data $(\boldsymbol{\alpha}_0,\boldsymbol{l}_0, \boldsymbol{\sigma}_0)$ assigned as above,  there is exactly one or zero RFFQM consistent with this data set.
\end{theorem*}
\noindent This theorem is established by a series of local calculations, starting at the panel on the bottom-left corner of the pattern (Fig.~\ref{fig:forward}(b)). The input data provides  the geometry and M-V assignment of creases  at three of four vertices of this panel. The  fourth vertex is then characterized by attempting to constrain the crease pattern to be rigidly and flat-foldable. The fundamental result derived in \cite{feng2020designs} is that, under this constraint, the  crease geometry and M-V assignment at this final vertex are either uniquely determined from the other three by explicit formulas, or the data is  incompatible\footnote{This means that there is no solution for the fourth vertex.}.   In the case of  compatible data, we can proceed to an adjacent panel, and iterate the calculation since we again know all relevant data at three of four vertices. The criteria for compatible data and the iterative formulas are provided in \ref{sect:ap-marching}.2-4.  By this procedure, we obtain an explicit marching algorithm that either discovers a unique RFFQM pattern or fails due to incompatibility at some point during iteration.

Let us assume compatible input data $(\boldsymbol{\alpha}_0,\boldsymbol{l}_0, \boldsymbol{\sigma}_0)$, so that we can compute the overall crease pattern by this marching algorithm.  This pattern is guaranteed to exhibit a single DOF folding motion that evolves {each folding angle from $0$ to $\pi$  (or $-\pi$)} monotonically under the prescribed M-V assignments.  We characterize this motion by a folding parameter $\omega$ such that $\omega = 0$ describes the flat crease pattern, $\omega = \pi$ the folded-flat state, and $0<\omega <\pi$ evolves the pattern from flat to folded flat (Fig.~\ref{fig:forward}(c)).  As a result, the kinematics of the origami structure are parameterized by
\begin{equation}
\begin{aligned}\label{eq:yij}
&\big\{ \mathbf{y}^{i,j}(\boldsymbol{\alpha}_0,\boldsymbol{l}_0, \boldsymbol{\sigma}_0, \omega) \big|  i = 0,1,\ldots, M, j = 0, 1, \ldots, N \big\},
\end{aligned}
\end{equation}
where $\mathbf{y}^{i,j}$ are the vertex positions (in 3D) on the deformed origami structure (determined by $\boldsymbol{\alpha}_0,\boldsymbol{l}_0, \boldsymbol{\sigma}_0, \omega$). This parameterization is also determined explicitly by marching \cite{feng2020designs}; the formulas for doing so are provided in \ref{sect:ap-marching}.5-7. { Briefly, regarding these formulas, the folding parameter $\omega$ corresponds to the folding angle of the fictitious crease below the $(0,0)$ vertex shown in Fig.~\ref{fig:forward}(a) and (d). This crease, along with the other boundary creases at the left and bottom boundary, are used to seed the pattern and its kinematics but are not included in the output  (Fig.~\ref{fig:forward}(c)) for the sake of convenience; their exclusion streamlines  the implementation of the overall marching algorithm.}

For inverse design, a key point with the marching algorithm is that discovering a RFFQM pattern and computing its kinematics in Eq.~(\ref{eq:yij}) are efficient calculations; both the pattern and any of its folded states are determined by computations that scale linearly with the number of panels as $ O(MN)$ due to the explicit iterative nature of the procedure.  Consequently, we can dedicate most of our computational resources towards solving an inverse design problem, rather than computing { the design and kinematic constraints involved in the discovery of such patterns.}

\subsection{General inverse design strategies.} 
Inverse problems in origami concern designing crease patterns to fold into structures with specified properties. Here, we describe a general framework for the inverse design of deployable origami structures by RFFQM, and discuss our approach in the context of related research. 

Let $\{ \mathbf{y}^{i,j}\}$ denote a collection of  vertices $\mathbf{y}^{0,0}, \mathbf{y}^{1,0}, \ldots ,\mathbf{y}^{M,N}$ in 3D space.  Suppose we want to arrange the vertices to correspond to an origami deformation of a RFFQM crease pattern with $M \times N$  panels that, in addition, {minimizes an objective function consistent with some inverse problem.} Tachi's Theorem \cite{ tachi2009generalization} indicates that we may describe { such a general optimization} in the following way:
{
\begin{equation}
\begin{aligned}\label{eq:asdf}
\min\limits_{\{ \mathbf{y}^{i,j}\}} &\quad f_{\text{obj.}} ( \{ \mathbf{y}^{i,j}\}) \\
\text{subject to}&\quad \begin{cases}
g_{\text{dev.}}(  \mathbf{y}^{i,j} , \mathbf{y}^{i+1,j},   \mathbf{y}^{i,j+1}, \mathbf{y}^{i-1,j}, \mathbf{y}^{i,j-1}  ) = 0  &\text{ if $(i,j)$ indexes an interior vertex},  \\
g_{\text{ffold.}}(  \mathbf{y}^{i,j} , \mathbf{y}^{i+1,j},   \mathbf{y}^{i,j+1}, \mathbf{y}^{i-1,j}, \mathbf{y}^{i,j-1}  ) = 0 &\text{ if $(i,j)$ indexes an interior vertex}.
\end{cases}
\end{aligned}
\end{equation}}
\noindent Here, $f_{\text{obj.}}$ corresponds to some suitably chosen objective function of all or some of the vertices $\{ \mathbf{y}^{i,j}\}$,  $g_{\text{dev.}}$ denotes the \textit{developability} constraint (the sector angles at a vertex sum to $2 \pi$), and $g_{\text{ffold.}}$ the \textit{flat-foldable} constraint (opposite sector angles at a vertex sum to $\pi$). As indicated, these constraints can be written in terms of five neighboring vertices on $\mathbb{R}^3$ via the formulas 
\begin{equation}
\begin{aligned}
&g_{\text{dev.}}(\mathbf{v}_0, \mathbf{v}_1, \mathbf{v}_2, \mathbf{v}_3, \mathbf{v}_4) = \sum_{i=1,\ldots,4} \arccos\big( \tfrac{\mathbf{v}_i - \mathbf{v}_0}{|\mathbf{v}_i - \mathbf{v}_0|} \cdot  \tfrac{\mathbf{v}_{i+1} - \mathbf{v}_0}{|\mathbf{v}_{i+1} - \mathbf{v}_0|} \big)  - 2\pi , \\
&g_{\text{ffold.}}(\mathbf{v}_0, \mathbf{v}_1, \mathbf{v}_2, \mathbf{v}_3, \mathbf{v}_4) = \arccos\big( \tfrac{\mathbf{v}_2 - \mathbf{v}_0}{|\mathbf{v}_2 - \mathbf{v}_0|} \cdot  \tfrac{\mathbf{v}_{1} - \mathbf{v}_0}{|\mathbf{v}_{1} - \mathbf{v}_0|} \big) + \arccos\big( \tfrac{\mathbf{v}_4 - \mathbf{v}_0}{|\mathbf{v}_4 - \mathbf{v}_0|} \cdot  \tfrac{\mathbf{v}_{3} - \mathbf{v}_0}{|\mathbf{v}_{3} - \mathbf{v}_0|} \big)  - \pi , \\
\end{aligned}
\end{equation} 
where $\mathbf{v}_5 = \mathbf{v}_1$ for the former formula, { and the side lengths $|\mathbf{v}_i -\mathbf{v}_0|$ are assumed\footnote{{ The assumption can be enforced by introducing inequality constraints to the optimization. Alternatively, a well-chosen objective function will ensure that optimization does not drive the design towards vanishing side lengths.}}  to be positive to apply the formulas.} More precisely, Tachi's theorem furnishes the following result: If we can find vertices $\{\mathbf{y}^{i,j}\}$ 
\begin{itemize}[leftmargin=*]
\item that solve all the $2(M-1)(N-1)$ equality constraints in Eq.~(\ref{eq:asdf}) for an $M \times N$ crease pattern
\item and that do not all lie { on a plane in} $\mathbb{R}^3$
\end{itemize}
then the set $\{\mathbf{y}^{i,j}\}$ describes the vertices of a rigid origami deformation of an $M \times N$ RFFQM crease pattern.  So any solution to Eq.~(\ref{eq:asdf}) corresponds to an origami structure that can  be designed on a flat reference crease pattern and deployed by a folding motion (mechanism) to achieve the objective.  Note though, regardless of the objective function, this optimization is a non-convex problem because the equality constraints are nonlinear. There are also $2(M-1)(N-1)$ nonlinear  equality constraints and $3(M+1)(N+1)$ unknowns  to optimize  in the collection of vertices $\{ \mathbf{y}^{i,j}\}$.  Thus, the optimization is  fundamentally challenging and fraught with potential scalability issues. Nevertheless, many researchers have attacked various aspects of this problem.

Tachi \cite{tachi2010freeform}, in particular, pioneered the approach in Eq.~(\ref{eq:asdf}) as a strategy for free-form origami by linearizing around  known origami patterns and choosing an objective that  perturbs the vertices towards desired positions, while maintaining the constraints.  Dudte et al.~\cite{Levi2016Programming} augmented the equality constraints in Eq.~(\ref{eq:asdf}) to allow some flexibility during the optimization and used this approach to demonstrate a variety of origami structures that approximate curved surfaces.  Importantly though, they achieved these results typically by relaxing  the flat-foldable constraint $g_{\text{ffold}}$; so the origami designs produced by their framework generally cannot deploy as strict mechanisms from the flat state to the optimized state. Hu et al.~\cite{Hu2020Rigid} recently adapted the  approach of  Dudte et al.~to allow for the flat-foldable constraint.  However, their method reported difficulty in converging to a solution for crease patterns with a large number of panels; an issue, we surmise,  is likely due to the significant challenge of satisfying all the equality constraints numerically when optimizing a finely meshed origami pattern.

Our key idea is to \textit{eliminate the equality constraints altogether} by taking advantage of the  characterization of RFFQM furnished by the marching algorithm. As discussed, this  algorithm efficiently parameterizes a RFFQM crease pattern by the angles, lengths and mountain-valley assignments on the ``L"-shaped boundary indicated by the arrays $(\boldsymbol{\alpha}_0, \boldsymbol{l}_0, \boldsymbol{\sigma}_0)$, and it efficiently parameterizes the kinematics by a folding parameter $\omega$. We can therefore directly and efficiently replace the optimization in Eq.~(\ref{eq:asdf}) with 

{
\begin{equation}
\begin{aligned}\label{eq:asdf1}
\min\limits_{\boldsymbol{\alpha}_0,\boldsymbol{l}_0, \boldsymbol{\sigma}_0, \omega} &\quad \tilde{f}_{\text{obj.}} (\boldsymbol{\alpha}_0, \boldsymbol{l}_0, \boldsymbol{\sigma}_0, \omega) \\
\text{subject to}&\quad \begin{cases}
(\boldsymbol{\alpha}_0, \boldsymbol{l}_0, \boldsymbol{\sigma}_0) \text{ is compatible input data}, \\
\omega \in (0, \pi).
\end{cases}
\end{aligned}
\end{equation}}
\noindent Here we simply replace  $f_{\text{obj.}}$ with $ \tilde{f}_{\text{obj.}} (\boldsymbol{\alpha}_0, \boldsymbol{l}_0, \boldsymbol{\sigma}_0, \omega) = f_{\text{obj.}}( \{ \mathbf{y}^{i,j}(\boldsymbol{\alpha}_0, \boldsymbol{l}_0, \boldsymbol{\sigma}_0, \omega)\})$; so it is the same general objective function, just with the vertices explicitly parameterized by the marching algorithm.

{While Eq.\;(\ref{eq:asdf1}) describes a general optimization scheme (equivalent to Eq.\;(\ref{eq:asdf})) over the family of RFFQM crease patterns, the  variable ${\boldsymbol{\sigma}}_0$ is an array of discrete variables and thus difficult to optimize numerically. As a point of practical implementation, we instead tackle this optimization problem under a prescribed $\bar{\boldsymbol{\sigma}}_0$, i.e., 
\begin{equation}
	\begin{aligned}\label{eq:asdf2}
	\min\limits_{\boldsymbol{\alpha}_0,\boldsymbol{l}_0, \omega} &\quad \hat{f}_{\text{obj}}(\boldsymbol{\alpha}_0, \boldsymbol{l}_0, \omega) =   \tilde{f}_{\text{obj.}} (\boldsymbol{\alpha}_0, \boldsymbol{l}_0, \bar{\boldsymbol{\sigma}}_0, \omega) \\
	\text{subject to}&\quad \begin{cases}
	(\boldsymbol{\alpha}_0, \boldsymbol{l}_0, \bar{\boldsymbol{\sigma}}_0) \text{ is compatible input data}, \\
	\omega \in (0, \pi).
	\end{cases}
	\end{aligned}
	\end{equation}
{ In this setting, the structure of compatible input data has nice properties for numerical implementation:} Suppose we have identified some compatible input data $(\bar{\boldsymbol{\alpha}}_0, \bar{\boldsymbol{l}}_0 , \bar{\boldsymbol{\sigma}}_0)$ (e.g., a well-known RFFQM origami structure like the Miura-Ori), then we can rigorously show that, with  M-V assignment indicated by $\bar{\boldsymbol{\sigma}}_0$ held fixed, there is an open neighborhood of $(\bar{\boldsymbol{\alpha}}_0, \bar{\boldsymbol{l}}_0)$ on which the data is also compatible. Additionally, we can show that the formulas for vertex positions $\bfy^{i,j}(\boldsymbol{\alpha}_0, \boldsymbol{l}_0, \bar{\boldsymbol{\sigma}}_0, \omega)$ are smooth for $(\boldsymbol{\alpha}_0, \boldsymbol{l}_0)$ in this neighborhood and for $\omega \in (0,\pi)$.
As a result, the optimization in Eq.~(\ref{eq:asdf2}) is a standard nonlinear programming problem over an open subset of $\mathbb{R}^{3M + 3N + 2}$ with smooth formula generating folded origami configurations, which means it can be treated using standard numerical schemes (see \ref{sect:ap-resources}).}
A precise statement and proof of this technical result is provided in \ref{sect:ap-marching}.8.

Finally, there are some notable benefits to formulating the optimization via {Eq.~(\ref{eq:asdf2})} rather than using vertex based approaches  like Eq.~(\ref{eq:asdf}):
\begin{itemize}[leftmargin=*]
\item \textit{Strict guarantees on deployability of the origami.}~Most of the striking examples of optimized origami structures found in the literature, such as in \cite{Levi2016Programming},  do not actually solve the optimization in Eq.~(\ref{eq:asdf}); rather, they solve an augmented vertex based approach that relaxes the flat-foldable constraints. In the relaxed setting, it is much easier to achieve an objective (e.g., surface approximation) since the optimization has less constraints. But there are also no guarantees that a functional property like {\textit{deployability}} --- the ability to fold as a mechanism to the target origami structure, either from the flat crease pattern or from a compact state --- can be ensured.  

In contrast, every origami structure obtained by the optimization in {Eq.~(\ref{eq:asdf2})} is rigidly and flat-foldable, so {deployable} in the sense described above. 

\item \textit{A direct encoding of the dimensionality of the design space.}~The design space for a RFFQM is proportional to the number of vertices on the boundary $O(M+N)$, not the number of vertices of the entire pattern $O(MN)$. This fact is directly encoded into the optimization in {Eq.~(\ref{eq:asdf2})}, yet not at all transparent\footnote{By counting DOFs and constraints in Eq.~(\ref{eq:asdf}), one might incorrectly  suspect that the design space in this optimization is in the order of $MN$.  This counting argument is, however, inappropriate because the equality constraints are nonlinear. As a simple example, the optimization problem ``$\min f(\mathbf{y})$ subject to  $|\mathbf{y}| = 0$" has a unique solution $\mathbf{y} = \mathbf{0}$ no matter the dimension of the ambient Euclidean space. { This solution always has zero DOF, even though the counting argument would suggest the DOFs of the ambient space.}} in the vertex based approach in Eq.~(\ref{eq:asdf}).

This disparity in dimensionality has significant implications for how one should design an objective function.  In particular, a common approach to constrained optimization is to introduce additional constraints that are consistent with the desired objective. In the optimization of surfaces developed in \cite{Levi2016Programming}, for example, the authors directly attach half of the vertices to the surface they wish to approximate, then attempt to use the remaining freedom to satisfy the origami constraints during the optimization. This procedure evidently works well when the flat-foldable constraints are relaxed. However, there are only $O(M + N)$ degrees-of-freedom in the design space of RFFQM. So any objective that introduces an additional set of $O(MN)$ nontrivial constraints to Eq.~(\ref{eq:asdf})  is {unlikely to be feasible} for $MN \gg1$.  
\end{itemize} 

\subsection{Optimizing for targeted surfaces}

\begin{figure}[t!]
\includegraphics[width=0.9\linewidth]{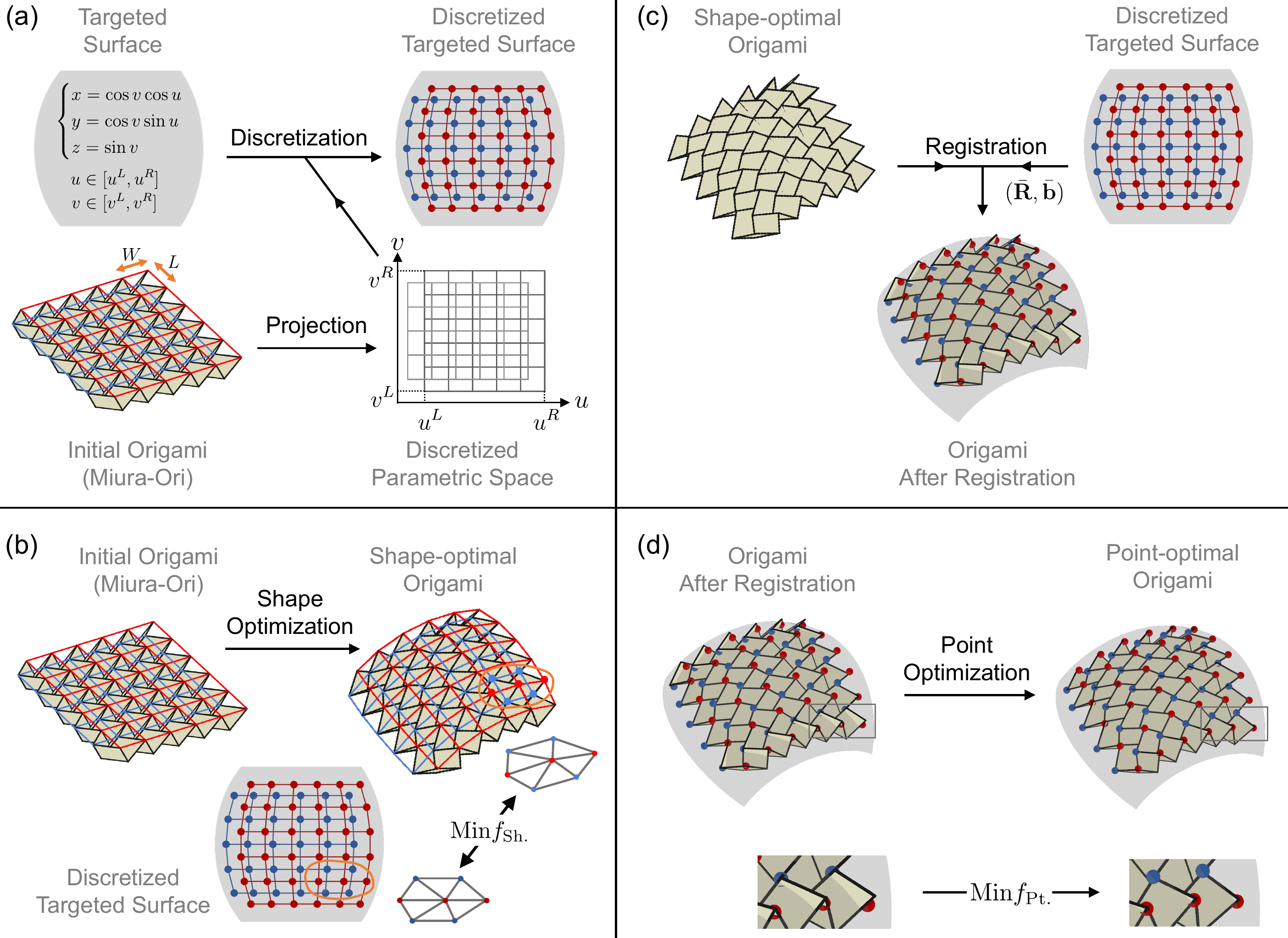}
\centering
\caption{{ Inverse design schematic. (a) Discretization: the targeted surface is discretized matching the grid number and offset of the initial Miura-Ori.  (b) Shape optimization: the targeted surface and origami surface are  triangulated, and discrete notions of metric and curvature are compared on like triangles to be optimized for shape. (c) Registration: the targeted surface is then rotated $\bar{\mathbf{R}}$ and translated $\bar{\mathbf{b}}$ onto the shape optimized patten in a process termed registration.  (d) Point optimization: a final optimization is performed to match the origami vertices to like vertices of the registered targeted surface to produce the optimal origami.}}
\label{fig:inverse}
\end{figure}

We now restrict our focus to a class of challenging and important inverse problems --- that of approximating targeted surfaces by optimizing over the family of deployable origami in Eq.~(\ref{eq:yij}).
We note however that the marching algorithm and the procedure outlined with {Eq.~(\ref{eq:asdf2})} can  be used for other optimization strategies, e.g.,  optimal packaging, locomotion, or
optimization of functional or dynamic properties.

In the problem of approximating surfaces using  deployable origami structures, we confront three basic issues:
\begin{enumerate}
\item[(I)]   Origami structures are rough, whereas the surfaces we often aim to approximate are smooth.
\item[(II)]   The delicate non-linear couplings relating $(\boldsymbol{\alpha}_0, \boldsymbol{l}_0, \boldsymbol{\sigma}_0, \omega)$ to the origami structure Eq.~(\ref{eq:yij}) and then to a surface, via some $\tilde{f}_{\text{obj.}}$ in {Eq.~(\ref{eq:asdf2})}, lead to an inverse design problem of minimizing a non-convex objective function over a non-convex set.
\item[(III)] Ensuring the functional property of deployability, while also obtaining a quality approximation, is challenging because the dimensionality of the design space is small, i.e., $O(M+N)$ as opposed to $O(MN)$.
\end{enumerate}
For these reasons, success in inverse design requires a careful strategy, both for formulating the optimization and choosing an initial condition.

We address these issues by embracing the Miura-Ori as a template for inverse design.  The key ideas are described in Fig.~\ref{fig:inverse}(a).  The initial origami shown is a Miura-Ori that has been folded along its single DOF motion to a 3D configuration, which we call a partially folded state (since it is neither flat, nor fully folded flat).  Importantly, this partially folded Miura-Ori, while itself a rugged corrugated structure, has an ordered collection of points forming red and blue lattices that discretize a planar region in 3D space.  Note, these lattices are offset from one another in the plane but have identical rectangular unit cells. Also, this basic fact holds regardless of the geometry/number of unit cells or the choice of  partially folded state.  We therefore take the offset lattice that emerges from a partially folded Miura-Ori as a seed to discretize the targeted surface and initialize the two-stage optimization that compares deployable origami structures to this surface.

{\textit{Basic setup with the Miura-Ori.}} To explain these ideas concretely, it is useful to describe the Miura-Ori using the marching algorithm. This is done by first choosing the input data to the algorithm: $(\boldsymbol{\alpha}_{\text{M-O}}, \boldsymbol{l}_{\text{M-O}}, \boldsymbol{\sigma}_{\text{M-O}})$ such that
\begin{equation}
\begin{aligned}\label{eq:inputData}
&\alpha^{i,0} =  \alpha,  && \beta^{i,0} = \pi - \alpha, && \text{for all $i$}, \\
&\alpha^{0,j} = \alpha, && \beta^{0,j} = \pi - \alpha,  && \text{for all } j \text{ even,} \\
&\alpha^{0,j} =  \pi - \alpha, && \beta^{0,j} = \alpha,  &&\text{for all } j \text{ odd,} \\
&w^{i,0} = w, &&   l^{0,j} = l,  && \text{for all } i \text{ and } j, \\
& \sigma^{i,0} = +, && \sigma^{0,j} = +, && \text{for all } i \text{ and } j,
\end{aligned}
\end{equation}
where $i$ and $j$ are cycled to give data consistent with an $M \times N$ pattern and  $0 <\alpha <\pi$ and $l,w > 0$ describe the geometry and corrugation of the origami.  We also assume $M$ and $N$ are even to simplify some notation below. This data, together with a folding parameter $0 < \omega_{\text{M-O}} < \pi$,  initializes the marching algorithm, which then produces a generic partially folded Miura-Ori with vertices
\begin{equation}
\begin{aligned}\label{eq:yMO}
\mathbf{y}^{i,j}_{\text{M-O}} = \mathbf{y}^{i,j}(\boldsymbol{\alpha}_{\text{M-O}}, \boldsymbol{l}_{\text{M-O}}, \boldsymbol{\sigma}_{\text{M-O}}, \omega_{\text{M-O}})
\end{aligned}
\end{equation}
in 3D space (recall Eq.~(\ref{eq:yij})).  This origami has an offset lattice, e.g., the red and blue mesh points for the initial origami in Fig.~\ref{fig:inverse}(a).  It is given by collecting the vertices  on the “top” surface of the origami through the formulas
\begin{equation}
\begin{aligned}\label{eq:rMO}
\mathbf{r}^{i,j}_{\text{M-O}} = \mathbf{y}^{2i,j}_{\text{M-O}}
\end{aligned}
\end{equation}
for $i =0,1, \ldots, M/2$ and $j = 0,1, \ldots, N$. ($M/2$ is an integer by assumption.)

Now, suppose we alter the angle and length input data in Eq.~(\ref{eq:inputData}) with perturbations $\boldsymbol{\alpha}_0 = \boldsymbol{\alpha}_{\text{M-O}} + \boldsymbol{\delta} \boldsymbol{\alpha}_0$ and $\boldsymbol{l}_0 =\boldsymbol{l}_{\text{M-O}} + \boldsymbol{\delta} \boldsymbol{l}_0$, while keeping the M-V assignment $\boldsymbol{\sigma}_{\text{M-O}}$ fixed.  A large class of these perturbations is compatible with RFFQM\footnote{The data is guaranteed to be compatible for sufficiently small perturbations. By our numerical investigation, it is also evident that the perturbations do not need to be all that small.}.  We can therefore initialize the marching algorithm, with a compatible perturbation and a folding parameter $0 < \omega < \pi$, to produce a new origami structure. As in the Miura-Ori case, we can collect the vertices on the  “top” surface of this new origami structure through the formulas
\begin{equation}
\begin{aligned}\label{eq:oriSurface}
\mathbf{r}^{i,j}(\boldsymbol{\alpha}_0, \boldsymbol{l}_0, \omega) = \mathbf{y}^{2i,j}(\boldsymbol{\alpha}_{0}, \boldsymbol{l}_{0}, \boldsymbol{\sigma}_{\text{M-O}}, \omega)
\end{aligned}
\end{equation}
for $i,j$ cycled as in Eq.~(\ref{eq:rMO}).  One way to view this collection is as a smooth deformation of the offset lattice in Eq.~(\ref{eq:rMO}). Embracing this viewpoint, we will call the mesh of these points an \textit{origami surface}. Note, each such surface is an explicit function of $(\boldsymbol{\alpha}_0, \boldsymbol{l}_0, \omega)$, and these  parameters can be varied.  So we can explore this large family of origami surfaces for the purpose of  inverse design.

Targeted surfaces of practical interest can often be described by a parameterization that maps a rectangular region in $2D$ to the surface.  Suppose we have one such surface  given by $\bar{\mathbf{r}}(u,v)$ for $u\in[u^L,u^R]$ and $v\in[v^L,v^R]$, e.g., the spherical cap in Fig.~\ref{fig:inverse}(a), and we wish to find an origami surface that resembles it. Since the origami surfaces above are inherently described by a discrete collection of points, we find it natural to make comparisons by invoking a discretization of the targeted surface given by
\begin{equation}
\begin{aligned}\label{eq:targetSurface}
\bar{\mathbf{r}}^{i,j} = \bar{\mathbf{r}}(u^{i,j}, v^{i,j})
\end{aligned}
\end{equation}
for $i,j$ cycled as in Eq.~(\ref{eq:rMO}). Here, the discrete points $u^{i,j}\in[u^L,u^R]$ and $v^{i,j}\in[v^L,v^R]$ are chosen based on a Miura-Ori offset lattice to exhibit the same zig-zag vertex distribution.  Recall that the offset lattice in Eq.~(\ref{eq:rMO}) depends on many parameters: $\alpha, w, l, N, M$ and $\omega_{\text{M-O}}$.  For simplicity and a uniform discretization, we choose $\alpha = \pi/3$, $w = l$ and $\omega_{\text{M-O}} = 3\pi/4$.  This choice results in an offset lattice with a nearly square unit cell of side lengths $L\approx W\approx l$, aspect ratio $W/L \approx 1$, total width $ \approx lM/2$, and total length $ \approx lN/2$. So we can treat the even integer $M$ as a free parameter dictating of the number of panels in the origami, then choose an even integer $N$ that best approximates the aspect ratio of the characteristic lengths\footnote{Here the characteristic lengths $\bar{L}_u$ and $\bar{L}_v$ represent the total size of the targeted surface along the $u$ and $v$ directions, respectively. For example we can take  $\bar{L}_u=u^R-u^L$ and $\bar{L}_v=v^R-v^L$ for the spherical cap in Fig.~\ref{fig:inverse}.} of the targeted surfaces $\bar{L}_u/\bar{L}_v \approx M/N$, and finally $l$ such that $lM/2 \approx \bar{L}_u$.  By these choices, we can use the construction in \ref{sect:ap-procedure}.1 to  project the offset lattice  to the $(u,v)$-plane, yielding a collection of points $(u^{i,j}, v^{i,j})$ that suitably discretize this space.

To this point, we have outlined a general strategy for obtaining a family of deployable origami surfaces and discretizing a (fairly) arbitrary targeted surface by choosing to embrace the Miura-Ori --- both for how we collect the points to describe origami surfaces and how we discretize the targeted surface.  These choices come with many benefits to the optimization, the heuristics of which are: 1)~The Miura-Ori is buried deep in the compatible set of parameters for RFFQM, meaning it can be perturbed in many directions without issues of incompatibility limiting the optimization. 2)~The red and blue mesh (Fig.~\ref{fig:inverse}(a)) collectively remains regular, even for large perturbations of a Miura-Ori. So there is a level of consistency when comparing these mesh points to analogous points on a smooth targeted surface.  3)~Finally, perturbed Miura-Ori surfaces have access to a wide range of effective curvatures and metrics.  So the optimization does not  get stuck in local minima of poor quality, at least for most surfaces of practical interest.  By combining these choices with a careful two-stage optimization procedure, we develop an approach that largely overcomes the issues discussed with (I-III).   We  develop the optimization procedure below, again using Fig.~\ref{fig:inverse} to guide the exposition. Additional details on the procedure are provided in \ref{sect:ap-procedure}.2.

{\textit{Step 1:~Discretization.}}~We fix a targeted surface of our choosing, $\bar{\mathbf{r}}(u,v)$, $u \in[u^L,u^R]$ and $v^{i,j}\in[v^L,v^R]$.  We also fix an even integer $M$, which sets the number of columns of panels for the origami. From these quantities, we construct the discretization of the targeted surface $\bar{\mathbf{r}}^{i,j}$ and the meshing of points of the origami surface $\mathbf{r}^{i,j}(\boldsymbol{\alpha}_0, \boldsymbol{l}_0, \omega)$ based on the Miura-Ori offset lattice, exactly as outlined with Eqs.~(\ref{eq:inputData}-\ref{eq:targetSurface}). For reference, we recall that the Miura-Ori parameters are labeled $(\boldsymbol{\alpha}_{\text{M-O}}, \boldsymbol{l}_{\text{M-O}}, \boldsymbol{\sigma}_{\text{M-O}}, \omega_{\text{M-O}})$; also, that the length input is chosen so that $w = l$ in Eq.~(\ref{eq:inputData}).  We therefore have $\boldsymbol{l}_{\text{M-O}} = l \mathbf{1}$ for an array $\mathbf{1}$, where each component is $1$.

\begin{figure}[t!]
\centering
\includegraphics[width=0.9\linewidth]{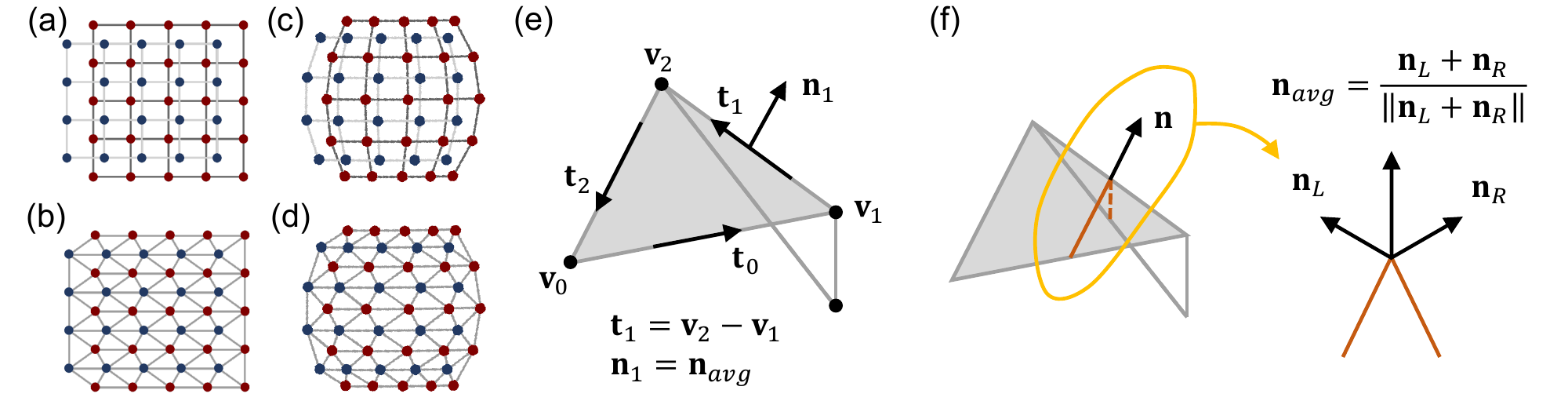}	
\caption{Shape description of the discrete surface: (a) Quad mesh in the parametric space; (b) Triangular mesh in the parametric space; (c) Quad mesh in the configuration space; (d) Triangular mesh in the configuration space. (e) Discretization of the tangent and normal vectors on the triangular mesh (reproduced from \cite{Rees2018Mechanics}). (f) Midedge normal (reproduced from \cite{Grinspun2006Computing}).}
\label{fig:discrete}
\end{figure}

{\textit{Step 2:~Shape optimization.}}  It is a well-known fact of differential geometry that two smooth parameterizations of surfaces from the same underlying domain are the same (up to  a Euclidean transformation) if and only if their first and second fundamental forms are the same. These geometric quantities are therefore the natural points of comparison for such parameterizations.    We consider an analogue of this comparison in an optimization of shape for triangular meshes of the targeted and origami surfaces. Specifically, we triangulate the two discretized surfaces ({far-right, Fig.~\ref{fig:inverse}(b)}) and compare shape operators that quantify discrete yet frame indifferent notions of metric and curvature on like triangles. 

The procedure we outline originates from \cite{Grinspun2006Computing} and has also been employed in the design of shape-changing biomimetic structures \cite{Rees2018Mechanics}.  As sketched in Figs.~\ref{fig:discrete}(a-d), we notice that the offset lattices and the discretization of surfaces (Eqs.~(\ref{eq:oriSurface}) and (\ref{eq:targetSurface})) based on these lattices have natural triangulations. Specifically, there are two sets of triangulations ---  one of the discrete origami surface in Eq.~(\ref{eq:oriSurface}) and one of the discrete targeted surface in Eq.~(\ref{eq:targetSurface}) --- with a one-to-one correspondence of triangles on each surface for which it is appropriate to compare shape.  We triangulate the two surfaces in this way and define shape operators on these triangles.  Consider a triangle on the interior of one of these surfaces and introduce the labeling in Figs.~\ref{fig:discrete}(e) and (f).  The shape operators for this triangle are 
\begin{equation}
\begin{aligned}
\mathbf{a} &= \frac{1}{\langle L \rangle}\left[\begin{matrix}  |\mathbf{t}_0| & |\mathbf{t}_1| & |\mathbf{t}_2| \end{matrix}\right], \\
\mathbf{b} &= 2\left[ \begin{matrix} \frac{(\mathbf{n}_1 - \mathbf{n}_0) \cdot \mathbf{t}_0}{|\mathbf{t}_0|}  &\frac{(\mathbf{n}_1 - \mathbf{n}_0) \cdot \mathbf{t}_1}{|\mathbf{t}_1|} &\frac{(\mathbf{n}_1 - \mathbf{n}_0) \cdot \mathbf{t}_2}{|\mathbf{t}_2|}  \\
\frac{(\mathbf{n}_2 - \mathbf{n}_1) \cdot \mathbf{t}_0}{|\mathbf{t}_0|}  &\frac{(\mathbf{n}_2 - \mathbf{n}_1) \cdot \mathbf{t}_1}{|\mathbf{t}_1|} &\frac{(\mathbf{n}_2 - \mathbf{n}_1) \cdot \mathbf{t}_2}{|\mathbf{t}_2|} \\
\frac{(\mathbf{n}_0 - \mathbf{n}_2) \cdot \mathbf{t}_0}{|\mathbf{t}_0|}  &\frac{(\mathbf{n}_0 - \mathbf{n}_2) \cdot \mathbf{t}_1}{|\mathbf{t}_1|} &\frac{(\mathbf{n}_0 - \mathbf{n}_2) \cdot \mathbf{t}_2}{|\mathbf{t}_2|} \end{matrix}\right],
\end{aligned}
\end{equation}
where $\mathbf{a}$ is non-dimensionalized by $\langle L \rangle$, the average length of the quad-mesh edges of the targeted surface (see \ref{sect:ap-procedure}.2).  In terms of the triangulation, the array $\mathbf{a}$ characterizes the shape of the triangle, and the matrix $\mathbf{b}$ describes its curvature since it relates to how this triangle is orientated relative to its neighbors. For completeness, the shape operators for a  boundary triangle take on a different form: we compute the array $\mathbf{a}$  as above but only compute one row of $\mathbf{b}$ since only two normals are defined in this case.

For the optimization, we list the local shape operators into a global ``shape array", so that each interior triangle contributes twelve elements to the list (the components of $\mathbf{a}$ and $\mathbf{b}$ above) and each boundary triangle six.  We let $\bar{\mathbf{S}}$ denote the shape array for the targeted surface and $\mathbf{S}(\boldsymbol{\alpha}_0, \boldsymbol{l}_0, \omega)$ the shape array for the origami surface.  We organize these shape arrays so that shape operator components of corresponding triangles --- on the origami and targeted surface --- have matching placement in these arrays.  With this organization, it is possible to show that $\bar{\mathbf{S}} = \mathbf{S}(\boldsymbol{\alpha}_0, \boldsymbol{l}_0, \omega)$ if and only if the two triangular meshes are the same up to Euclidean transformation; hence, the connection to first and second fundamental forms.  We therefore introduce
\begin{equation}
\begin{aligned}\label{eq:fSh}
f_{\text{Sh.}}( \boldsymbol{\alpha}_0, l_0 , \omega) = \frac{1}{N_T} |\bar{\mathbf{S}} - \mathbf{S}(\boldsymbol{\alpha}_0, l_0 \boldsymbol{1}, \omega)|^2
\end{aligned}
\end{equation}
as the objective function if $(\boldsymbol{\alpha}_0, l_0 \boldsymbol{1}, \boldsymbol{\sigma}_{\text{M-O}})$ is compatible input data to the marching algorithm, where $N_T$ is the number of vertices on the origami surface. If the data is, instead,  incompatible, a large positive constant $C_{\text{Num.}} \gg  f_{\text{Sh.}}(\boldsymbol{\alpha}_{\text{M-O}},l, \omega_{\text{M-O}})$ is returned.  This is a numerically convenient way to enforce compatibility during the optimization.   Note, the length input data here, $\boldsymbol{l}_0 = l_0 \boldsymbol{1}$, is restricted to be a rescaling of the Miura-Ori length input in order to preserve a relatively uniform aspect ratio for the origami panels.  To compute an optimum, we start from the Miura-Ori input data $(\boldsymbol{\alpha}_{\text{M-O}}, l, \omega_{\text{M-O}})$ and iterate numerically (see \ref{sect:ap-resources}) to arrive at a local minimum for $f_{\text{Sh.}}(\cdot)$ or an origami configuration which lies near the boundary of the compatible set of RFFQM input parameters.  For reference, we label this optimum $( \boldsymbol{\alpha}^{\star}_0, l^{\star}_0 , \omega^{\star})$. {The overall procedure is sketched in Fig.\;\ref{fig:inverse}(b).}

{\textit{Step 3:~Registration.}}~The origami surface obtained from shape optimization is indicated by the collection of vertices $\mathbf{r}^{i,j}(\boldsymbol{\alpha}^{\star}_0, l_0^{\star} \boldsymbol{1} , \omega^{\star})$.  Since shape optimization is frame indifferent, these vertices are not necessarily aligned and oriented with the like vertices $\bar{\mathbf{r}}^{i,j}$ on the targeted surface. Thus, we apply a rigid motion to the targeted surface to fit the vertices as best as possible by solving
\begin{equation}
\begin{aligned}
\min_{\mathbf{R} \in SO(3), \mathbf{b} \in \mathbb{R}^3} \sum_{i,j} | \mathbf{r}^{i,j}(\boldsymbol{\alpha}^{\star}_0, l_0^{\star} \boldsymbol{1} , \omega^{\star}) - (\mathbf{R} \bar{\mathbf{r}}^{i,j} + \mathbf{b})|^2
\end{aligned}
\end{equation}
({Fig.~\ref{fig:inverse}(c)}). Since the minimizing rigid motion here can be large, the solution is computed numerically using the coherent point drift method (see \ref{sect:ap-resources}), which is based on well-established ideas \cite{Myronenko5432191}.  For reference, we label the minimizing pair $(\bar{\mathbf{R}}, \bar{\mathbf{b}})$.

{\textit{Step 4:~Point optimization.}}~For many targeted surfaces, shape optimization provides a reasonable global approximation of shape.  However, there can be significant local deviation ({e.g., near the boundary of the pattern after registration; Fig.~\ref{fig:inverse}(c)}). We improve the approximation by perturbing the parameters $(\boldsymbol{\alpha}^{\star}_0, l_0^{\star} \boldsymbol{1} , \omega^{\star})$ to bring like vertices on the origami and targeted surface closer together ({Fig.~\ref{fig:inverse}(d)}).  Specifically, we introduce the second optimization step that takes
\begin{equation}
\begin{aligned}
&f_{\text{Pt.}}(\boldsymbol{\alpha}_0,\boldsymbol{l}_0, \omega, \bfR, \bfb) =   \frac{1}{N_T} \sum\limits_{i,j}| \mathbf{r}^{i,j}(\boldsymbol{\alpha}_0, \boldsymbol{l}_0, \omega)  -(\mathbf{R}\bar{\mathbf{r}}^{i,j} + \mathbf{b}) |^2
\end{aligned}
\end{equation}
as the objective function for compatible input data, and a large number  $\tilde{C}_\text{Num.} \gg f_{\text{Pt.}}( \boldsymbol{\alpha}_0^{\star},  l_0^{\star} \boldsymbol{1}, \omega^{\star}, \bar{\mathbf{R}}, \bar{\mathbf{b}})$ when the data is incompatible.  Note, the full set of length input data $\boldsymbol{l}_0$ is freely optimized in this step, and we also include a rigid motion term for a rotation $\mathbf{R} = \bar{\mathbf{R}} + \boldsymbol{\delta} \mathbf{R}$ and translation $\mathbf{b} = \bar{\mathbf{b}} + \boldsymbol{\delta}\mathbf{b}$ (likely optimal as small perturbations of the motion in registration).  Note also, in performing the optimization, we parameterize the rotation fully in terms of Euler angles by writing $\mathbf{R} = \mathbf{Q}(\eta, \xi, \zeta) \bar{\mathbf{R}}$ for 
\begin{equation}
\mathbf{Q}(\xi,\eta,\zeta)=
\left[ \begin{matrix} 
\cos\xi\cos\zeta-\cos\eta\sin\xi\sin\zeta  &-\cos\xi\sin\zeta-\cos\eta\cos\zeta\sin\xi &\sin\xi\sin\eta  \\
\cos\zeta\sin\xi+\cos\xi\cos\eta\sin\zeta  &\cos\xi\cos\eta\cos\zeta-\sin\xi\sin\zeta  &-\cos\xi\sin\eta \\
\sin\eta\sin\zeta                          &\cos\zeta\sin\eta                          &\cos\eta
\end{matrix}\right],
\end{equation}
where $\eta, \xi, \zeta$ are the Euler angles.  Finally, to compute an optimum, we start from the shape optimized  input data $(\boldsymbol{\alpha}_0^{\star}, l_0^{\star} \boldsymbol{1}, \omega^{\star})$, rotation $\bar{\mathbf{R}} = \mathbf{Q}(0,0,0) \bar{\mathbf{R}}$, and translation $\bar{\mathbf{b}}$. Then, we iterate numerically in the same manner as shape optimization (see \ref{sect:ap-resources}).  For reference, we label the optimal input data as $(\boldsymbol{\alpha}^{\star \star}_0, \boldsymbol{l}_0^{\star \star} , \omega^{\star \star})$, the optimal rotation as $\bfR^{\star \star}= \mathbf{Q}(\eta^{\star \star}, \xi^{\star \star}, \zeta^{\star \star}) \bar{\mathbf{R}}$, and the optimal translation as $\bfb^{\star \star}$.

{\textit{Step 5:~Quality of approximations.}}~We measure the maximum distance between like vertices on the origami and targeted surfaces after registration to characterize the quality of approximation.  This calculation is done for both the shape optimized and point optimized surfaces, i.e.,
\begin{equation}
\begin{aligned}
&d^{\star}= \max_{i,j} \frac{1}{\langle L \rangle} |\mathbf{r}^{i,j}(\boldsymbol{\alpha}^{\star}_0, l_0^{\star} \boldsymbol{1},  \omega^{\star}) -(\bar{\mathbf{R}} \bar{\mathbf{r}}^{i,j} + \bar{\mathbf{b}}) |, \\
&d^{\star \star}= \max_{i,j} \frac{1}{\langle L \rangle} |\mathbf{r}^{i,j}(\boldsymbol{\alpha}^{\star \star}_0, \boldsymbol{l}^{\star \star}_{0},  \omega^{\star\star}) - (\mathbf{R}^{\star \star} \bar{\mathbf{r}}^{i,j} + \mathbf{b}^{\star \star}) |,
\end{aligned}
\end{equation}
respectively. Here, $\langle L \rangle$ denotes the average length of the quad-mesh edges discretizing the targeted surface (see \ref{sect:ap-procedure}.2).  

{\textit{Comments and generalizations.}} {The heuristic for success in this two-stage optimization procedure is that shape optimization does the bulk of the work --- approximating well the global surface, except possibly in some small regions of the pattern ---  while point optimization supplies a refinement that corrects the poorly approximated regions but otherwise does not dramatically change the pattern.} Recall that the input length array is restricted when optimizing for shape (Eq.~(\ref{eq:fSh})).  Without this restriction, we find shape optimization to be far too flexible, leading to origami with distorted aspect ratios ill-suited for practical application (see \ref{sect:ap-procedure}.2 and Fig.~\ref{fig:ap-distort}).  We also find point optimization to be delicate, leading to quality results only if the input is already fairly close to the desired surface.  { These features motivated our two-stage optimization procedure.}  We typically look for shape optimization to yield $d^{\star} < 1$, so that distances between like vertices on the two surfaces are no larger than the characteristic length of the mesh-panels.  In this case, point optimization often yields quality refinement $d^{\star \star} \approx 0.5$ to 0.05 $ d^{\star}$ (see Table~\ref{tabel:eval}) by slight perturbation $(\boldsymbol{\alpha}^{\star \star}_0, \boldsymbol{l}^{\star \star}_{0},  \omega^{\star\star}) \approx (\boldsymbol{\alpha}^{\star}_0, l_0^{\star}\boldsymbol{1},  \omega^{\star})$, a result likely facilitated by the extreme non-linearity inherent to folding origami.

Finally, for some targeted surfaces discussed below, we choose an initial origami different from the Miura-Ori, as it significantly improves the quality of approximation resulting from the optimization.  We base our choice on direct numerical observations of origami structures that exhibit useful basic deformations (see \ref{sect:ap-observations}.1 and Fig.~\ref{fig:ap-deformation}).  { The versatility of our framework allows us to simply apply the step-by-step inverse design procedure as before, except replacing the Miura-Ori with the new initial origami and modifying the offset lattice that discretizes the targeted shape.} As the input origami need not be consistent with a planar or uniform discretized surface, we choose this offset lattice to coincide with the average tangent space of this surface (see \ref{sect:ap-observations}.2 and Fig.~\ref{fig:ap-offset}).

\begin{figure}[t!]
	\includegraphics[width=0.9\linewidth]{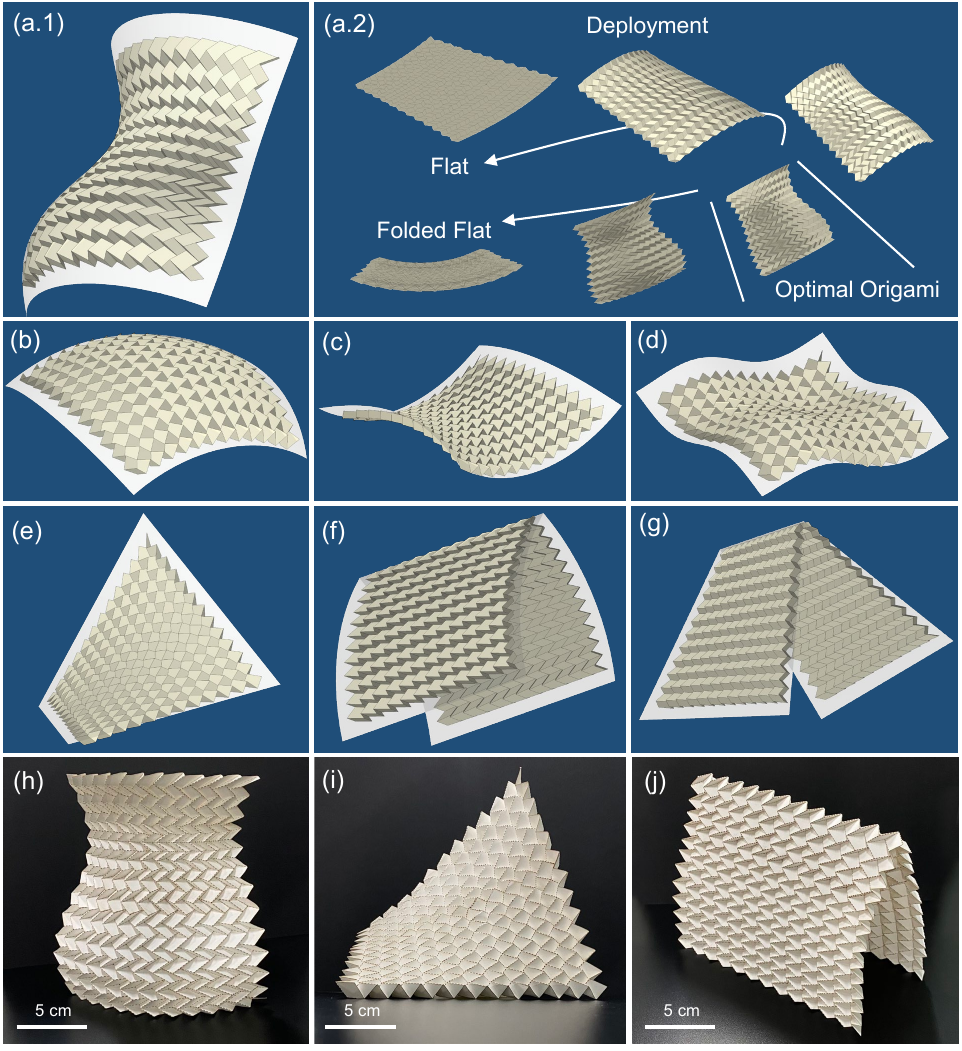}
	\centering
	\caption{Inverse design of surfaces by deployable origami. { Numerical examples:} (a) A quarter vase, and its deployment from states that are easy to manufacture and package. (b-e) Other examples of smooth surfaces: (b) Spherical cap, (c) hyperboloid, (d) 2D sinusoid, (e) saddle.  (f-g) Examples of surfaces with sharp ridges: (f) Connecting cylinders and  (g) connecting saddles. { Paper models: (h) The quarter vase. (i) The saddle. (j) The connecting cylinders.} }
	\label{fig:examples}
\end{figure}

\section{Examples and discussion}

\subsection{Numerical examples of deployable origami}
\label{sect:deployableExamples}
{ We demonstrate the numerical examples of our inverse design framework for targeted surfaces in Figs.~\ref{fig:examples} (a-g).}  In each example, the targeted surface is overlaid in grey onto the deformed optimal origami structure.  To reiterate, the origami structures shown are deployable --- they can be obtained by a single DOF folding motion from an easily manufactured flat state and an easily packaged folded-flat state.  We display this deployment capability with the first example and refer to Fig.~\ref{fig:ap-examples} and the supplementary Videos~S1-S7 for the others.  
With our marching algorithm, the optimization process takes a matter of minutes using standard computational resources (see \ref{sect:ap-resources}) for various given examples. We also numerically investigate the efficiency of our optimization scheme. The results of this investigation are provided in Fig.~\ref{fig:time}, and they indicate quadratic time complexity in terms of the number of origami panels. Finally, Table~\ref{tabel:eval} shows the exact  parameterizations of the targeted surfaces, i.e., $\bar{\mathbf{r}}(u,v)$, details about the initial Miura-Ori and optimal origami, and computational time for the optimization. 

\begin{figure}[!t]
\centering
\includegraphics[width=0.8\linewidth]{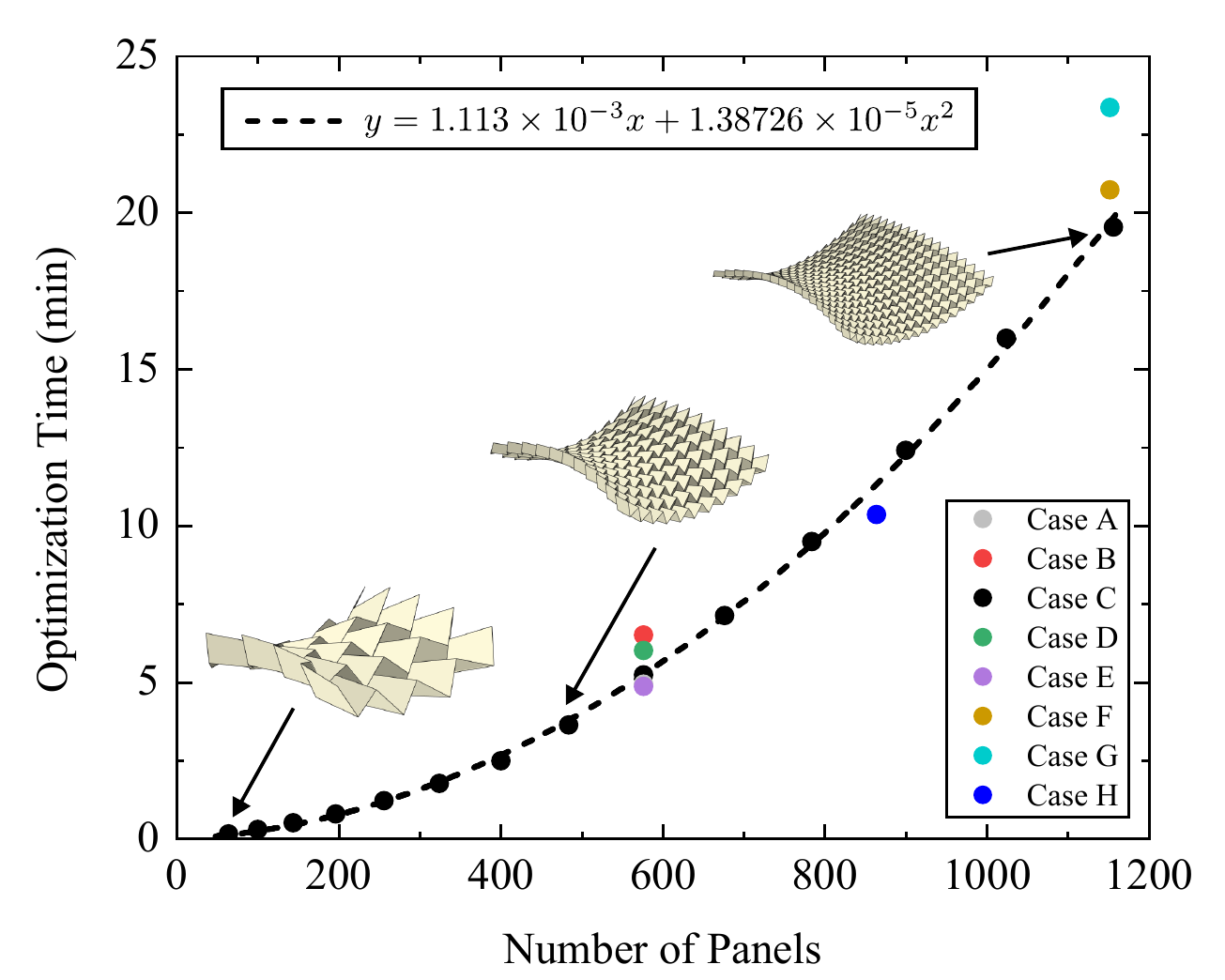}
\caption{Time consumption of the numerical cases. The data points of case C are obtained from a series of approximations of the same hyperboloid targeted surface (C in Table~\ref{tabel:eval}) with different origami size $M\times N=8^2,10^2,12^2,\cdots,34^2$. Other data points are plotted according to the computational time for the simulations for A-H in Table~\ref{tabel:eval}. The polynomial fit of the data points of case C gives a quadratic curve with the coefficient of determination $R^2=0.99972$. Points of the other cases lie near the fitting-curve, suggesting that the optimization approach we provide has $O((MN)^2)$ time complexity for various targeted surfaces.  }
\label{fig:time}
\end{figure}

\begin{table}[t!]
\centering
\begin{minipage}{0.7\textwidth}
\centering
\caption{Evaluation of Approximation}
\label{tabel:eval}
\resizebox{\columnwidth}{!}{
\begin{tabular}{llccccc}
No. &Targeted surface & Initial Origami \footnote{The notation employed to describe the initial origami crease pattern, i.e., $P_{24,24}^{\rm{mu}}\left( \pi/3,2 \pi/3\right),\ldots,$ etc., is defined explicitly in \ref{sect:ap-observations}.1.}& $\omega^{\star\star}$ & $d^{\star}$ & $d^{\star\star}$ & Time\\
\midrule
A & $\begin{aligned}&\left\{
\begin{aligned}
x & = (0.75+0.12\sin{\pi v})\cos{u}\\
y & = (0.75+0.12\sin{\pi v})\sin{u}\\
z & = 0.7v
\end{aligned}
\right.\\&~~~~u\in[0,0.5\pi],~v\in[0.2,2.0]\end{aligned}$ & 
\makecell[c]{$P_{24,24}^{\rm{mu}}\left(\pi/3,2\pi/3\right)$\\$l=w=0.1$\\$\omega=0.75\pi$} & $0.7544\pi$& $0.6224$& $0.1351$& $296.01s$\\
\midrule
B & $\begin{aligned}&\left\{
\begin{aligned}
x & = \cos{u}\cos{v}\\
y & = \sin{u}\cos{v}\\
z & = \sin{v}
\end{aligned}
\right.\\&~~~~u\in[-\pi/6,\pi/6],~v\in[-\pi/6,\pi/6]\end{aligned}$ & 
\makecell[c]{$P_{24,24}^{\rm{mu}}\left(\pi/3,2\pi/3\right)$\\$l=w=0.075$\\$\omega=0.75\pi$} & $0.7946\pi$& $0.7006$& $0.1071$& $390.75s$\\
\midrule
C & $\begin{aligned}&\left\{
\begin{aligned}
x & = -\cos{u}\sqrt{1+v^2}\\
y & = -\sin{u}\sqrt{1+v^2}\\
z & = -v
\end{aligned}
\right.\\&~~~~u\in[-\pi/6,\pi/6],~v\in[-0.5,0.5]\end{aligned}$ & 
\makecell[c]{$P_{24,24}^{\rm{mu}}\left(\pi/3,2\pi/3\right)$\\$l=w=0.075$\\$\omega=0.75\pi$} & $0.7139\pi$& $0.2714$& $0.1210$& $313.75s$\\
\midrule
D & $\begin{aligned}&\left\{
\begin{aligned}
x & = u\\
y & = v\\
z & = (1+0.12\sin{\pi u})(1+0.12\sin{\pi v})
\end{aligned}
\right.\\&~~~~u\in[0,2],~v\in[0,2]\end{aligned}$ & 
\makecell[c]{$P_{24,24}^{\rm{mu}}\left(\pi/3,2\pi/3\right)$\\$l=w=0.15$\\$\omega=0.75\pi$} & $0.7463\pi$& $0.3114$& $0.0995$& $360.85s$\\
\midrule
E & $\begin{aligned}&\left\{
\begin{aligned}
x & = u\\
y & = v\\
z & = uv
\end{aligned}
\right.\\&~~~~u\in[-0.5,0.5],~v\in[-0.5,0.5]\end{aligned}$ & 
\makecell[c]{$P_{24,24}^{\rm{mu}}\left(\pi/3,2\pi/3\right)$\\$l=w=0.075$\\$\omega=0.75\pi$} & $0.7723\pi$& $0.3863$& $0.1061$& $291.97s$\\
\midrule
F & $\begin{aligned}&\left\{
\begin{aligned}
x & = u\\
y & = -{\rm{sign}}v(\cos(|v|-\pi/4)-\sqrt{0.5})\\
z & = \sin(|v|-\pi/4)+\sqrt{0.5}
\end{aligned}
\right.\\&~~~~u\in[-0.5,0.5],~v\in[-\pi/4,\pi/4]\end{aligned}$ & 
\makecell[c]{$P_{24,48}^{\rm{pl}}\left(\pi/3,115\pi/180\right)$\\$\sigma^{0,25}=-1$\\$l=w=0.075$\\$\omega=0.75\pi$} & $0.7759\pi$& $0.5460$& $0.1142$& $1244.05s$\\
\midrule
G & $\begin{aligned}&\left\{
\begin{aligned}
x & = u\\
y & = v/2+\sqrt{3}u|v|/6\\
z & = \sqrt{3}|v|/2+uv/6
\end{aligned}
\right.\\&~~~~u\in[-0.5,0.5],~v\in[-1,1]\end{aligned}$ & 
\makecell[c]{$P_{24,48}^{\rm{pl}}\left(\pi/3,115\pi/180\right)$\\$\sigma^{0,25}=-1$\\$l=w=0.075$\\$\omega=0.75\pi$} & $0.7493\pi$& $1.6267$& $0.1194$& $1401.69s$\\
\midrule
H & \makecell[l]{Discrete points in Fig.~\ref{fig:face}(a).} & 
\makecell[c]{$P_{36,24}^{\rm{vb}}\left(\pi/2,3\pi/4\right)$\\$l=w=2/30$\\$\omega=7\pi/9$} & $0.7896\pi$& $0.6761$& $0.5390$& $621.15s$\\
\midrule
I & Same as that in case H. & 
\makecell[c]{$P_{36,24}^{\rm{vb}}\left(\pi/2,3\pi/4\right)$\\$l=w=2/30$\\$\omega=7\pi/9$} & $0.7624\pi$& $0.2407$& $0.1292$& \makecell[c]{$23729.42s$\\ \text{(on HPC)}}\\
\bottomrule
\end{tabular}}
\end{minipage}
\end{table}

\subsubsection{Surfaces with modest curvature}
With the examples in Figs.~\ref{fig:examples}{(a-e)}, we approximate a variety of smooth surfaces  by deployable origami. Each example shown is obtained by optimization, starting from the same $24 \times 24$ partially folded Miura-Ori, and following the inverse design framework above exactly.  We choose the targeted surfaces here --- a quarter vase, spherical cap, hyperboloid, 2D sinusoidal parameterization and saddle ---  to demonstrate the wide range of curvatures amenable to our methods.  Each surface is approximated with the value of $d^{\star \star} \sim 0.1$. Note, we did choose the curvatures  to be modest compared to the size of the mesh-panels and to not exhibit dramatic variations.   We elaborate more on this below when discussing the human face case.  Note also, each optimization here took $\approx 5$ minutes using the modest computational resources outlined in \ref{sect:ap-resources}.

\subsubsection{Surfaces with sharp ridges}

The examples  of smooth surfaces above arise by taking a partially folded Miura-Ori  as the initial state to the optimization.  Yet, our marching algorithm for deployable origami is general. So any base state with compatible input parameters $(\boldsymbol{\alpha}_0, \boldsymbol{l}_0, \boldsymbol{\sigma}_0)$ can be used to seed the inverse design framework. To explore this idea, we first note that there is a family of input parameters more general than those of the Miura-Ori that produce crease patterns with periodicity.  These patterns deform periodically when folding along the Miura-Ori M-V assignment.  However, it is possible to change the M-V assignment $\boldsymbol{\sigma}_0$, keeping the other parameters fixed (see \ref{sect:ap-observations}.3 and Fig.~\ref{fig:ap-deformation}(a)).  When we flip exactly one of these assignments, the marching algorithm produces exactly the same periodic crease pattern, but it folds as two planar states connected by something akin to a sharp interface at this altered assignment --- a feature reminiscent of the multi-stability of flat crease patterns explored in \cite{dieleman2020jigsaw}. 

We take advantage of this fact to explore the inverse design of targeted surfaces with sharp interfaces. With  Figs.~\ref{fig:examples}{(f) and (g)}, we consider two such examples: one connecting cylinders and the other connecting saddles.  In both cases, the initial origami to the optimization is as described above, i.e., a periodic origami with $24\times48$ mesh-panels and an altered M-V assignment along the $25$th column that produces a sharp interface in its folded states.  The origami obtained by optimization well approximate the surfaces, $d^{\star \star} \sim 0.1$, even with the sharp interfaces.  An interesting point highlighted by this result is that we are not limited by smoothness.  Structured triangulations are obtained and compared for both the targeted and origami surfaces in shape optimization.   Whether or not the targeted surface is smooth, we have the property $f_{\text{Sh.}}(\cdot) =0$ in Eq.~(\ref{eq:fSh}) if and only if the  two triangulations are the same up to rigid motion (as discussed above).  Since the triangulations serve as the fundamental proxy of the surfaces, the framework itself addresses the issue of smoothness automatically.

\subsubsection{Paper models.}

{ To verify our numerical results, we fabricated and folded the designs for the quarter vase, the saddle, and the connecting cylinders.  For each example, a flat piece of cardboard 315g art paper was perforated by a laser cutter in accordance with the crease pattern design generated by the simulation. The paper models were then folded manually in an attempt to match the desired target shape. The folded shapes are shown in Figs.~\ref{fig:examples} (h-j). One can observe that these shapes agree well with the numerical examples in Figs.~\ref{fig:examples} (a), (e), and (f), respectively.}

\begin{figure}[t!]
\includegraphics[width=0.9\linewidth]{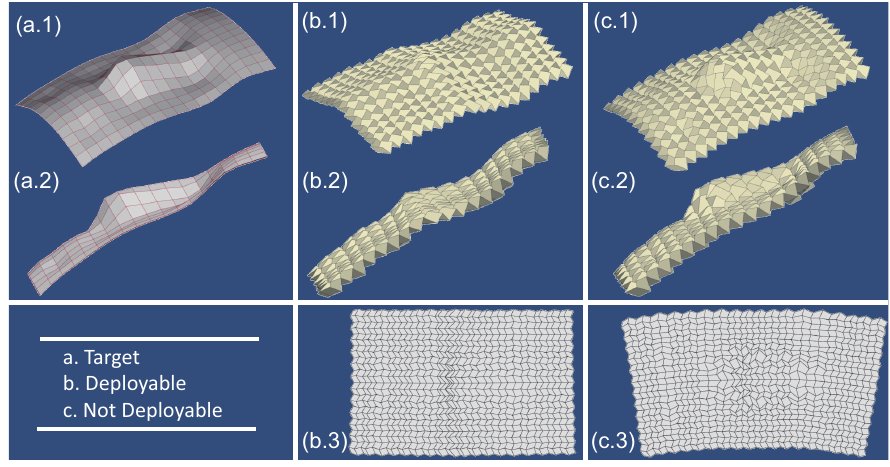}
\centering
\caption{Making a face with quad-mesh origami: (a) Targeted surface from two perspectives. (b) Attempt at making a face using RFFQM.  This origami  can be folded from the flat crease pattern to the ``face'' depicted without incurring stress in the mesh-panels.  (c) The approximation can be improved significantly using quad-mesh origami that is not constrained by deployability.  This origami is ideally  stress-free in both the flat and face state but can only be deployed by stressing the panels. The regular distributed discrete human-face points is reproduced from the source data in \cite{Thomas2018Morphable} by interpolation and transformation.}
\label{fig:face}
\end{figure}

\subsection{Challenges in deployable origami}

Ever since the pioneering work  \cite{klein2007shaping} on shape-programming with hydrogels, it has been an ambition to develop design principles for  programming material physics to actuate completely general surfaces \cite{dias2011programmed, plucinsky2016programming, van2017growth,aharoni2018universal,griniasty2019curved}.  In this domain, the challenge of ``making a face'' --- with its diverse and sharp changes in curvature --- is seen as a worthy exemplar of a general surface.  We therefore consider this case in the context of quad-mesh origami (Fig.~\ref{fig:face}).

Before turning to the results, we digress to lend perspective on an important issue. {The  state of computational origami design for artistic purposes}  is impressive \cite{lang1996computational,demaine2017origamizer}, showcasing an ability to approximate complex shapes (insects/bunnies/\ldots) far beyond anything we could ever hope to approach by our methods.   Rigorous mathematical results \cite{conti2008confining} also suggest that essentially any surface can be approximated by origami maps.  Nevertheless, both the art and the mathematics take advantage of refinement; a base pattern gets enriched with finer and finer folding mechanisms that are guaranteed to approximate a targeted shape to any degree of accuracy with sufficient refinement.  By its very nature, this strategy is a direct impediment to considerations of manufacturability or foldability.  In contrast, we are seeking to approximate a surface by an origami that is guaranteed to be deployable.  This distinction is important in engineering: Folding the crease pattern to achieve such shape is not the work of a skilled Origamist, but rather modalities like simple external loads or actuators involving motors/active materials,  which have practical limitations.

{Unfortunately, not all surfaces can be easily approximated using our optimization strategy for deployable origami.} The face we attempt to approximate is shown in Fig.~\ref{fig:face}(a).   Our first attempt using a partially folded Miura-Ori  ($36 \times 24$ mesh-panels) as the initial origami failed to produce a quality approximation during the optimization.  Our best attempt is shown in Fig.~\ref{fig:face}(b).  To obtain this approximation, we employed an initial origami exhibiting a basic deformation --- slight curvature along the short axis ---  that most resembled the global shape of the face (Table \ref{tabel:eval}).  After optimizing through our design framework, we end with a deployable origami structure that fits the global features of the face but fails to capture the sharp features at the nose.  To understand why, note the crease pattern for this origami (Fig.~\ref{fig:face}(b.3)).  The mark of being on the boundary of the compatible set of RFFQM is diminishing panel lengths or aspect ratios.  The central region of the pattern, i.e., the region  folded to approximate the nose, fits exactly this description.  So the optimization is driving the initial origami towards the boundary of the compatible set of parameters, where it gets stuck before the nose can be adeptly approximated.

\subsection{More general quad-mesh origami} 

What if we relax the condition of  flat-foldability, but otherwise maintain that the origami is folded from a flat quad-mesh crease pattern? One generally expects enhanced flexibility in the parameter space, thus a richer capability to approximate general surfaces. We have  generalized our inverse design framework to explore this question.  The basic  approach is a local-to-global characterization that yields an explicit parameterization of the origami, akin to the marching algorithm for Eq.~(\ref{eq:yij}), except with significantly more degrees-of-freedom in the angle input data.  It is possible to show that, unlike the flat-foldable case, every interior vertex in this setting has a single degree of freedom in its choice of sector angles.  If these are chosen appropriately and if appropriate input data on the left and bottom boundary of the pattern is supplied, then the entire origami structure can be determined explicitly by marching. A related approach was derived recently in \cite{dudte2020additive}.

In this more general framework, we can take our parameterization for quad-mesh origami and apply exactly the same optimization procedure above to investigate the inverse design of surfaces.   Our effort to make a face in this setting proved quite successful (Fig.~\ref{fig:face}(c)).  With the same input origami as the one supplied to the optimization in Fig.~\ref{fig:face}(b), we are able to adeptly approximate the face.  { The nose, which caused so much trouble previously, now emerges from a quite different crease pattern (Fig.~\ref{fig:face}(c.3)) that exhibits nothing of the characteristic zig-zags inherent to perturbed Miura-Ori.}  The optimization is clearly exploiting the additional degrees of freedom at interior vertices to address the sharp contrasts in curvature.

{ Fig.~\ref{fig:face} shows some challenges inherent to the constraint of deployability:} Two sharply distinctive origami  emerge from the same initial origami under the same optimization framework, except one is deployable and the other is  not.  To be clear on the latter, the origami in Fig.~\ref{fig:face}(c) is (ideally) stress-free in both the flat state and face state, but it cannot be folded without introducing stress in the panels during the process.  Whether such stresses can be overcome by the simple modalities of folding inherent to practical engineering depends on additional factors, such as the stiffness of panels and the actuation strategy. Two improvements to our ideas can be made on the deployability front: 1) All the initial origami to the optimization are essentially derivatives of the Miura-Ori.  This is not necessary.
{ Some well-designed non-Miura-Ori-like quadrilateral mesh patterns have numerous ways of folding with distinct folded configurations \cite{Liu2021Origami}.
Due to their versatility, we expect that these foldable families, when utilized as input to the optimization algorithm, can help to enlarge the design space of surfaces that can be adeptly approximated.}  2) A general characterization of the foldability of quad-meshes is provided in \cite{izmestiev2017classification}, yet the characterization is hard to ``march'' due to non-local couplings. If this difficulty can be overcome, it may be possible to give up on flat-foldability --- which is not needed in some fields --- without giving up deployability.

\section{Conclusion} 

In this paper, we demonstrate an inverse design framework that is easy to implement, efficient, and accurate for approximating targeted surfaces by deployable origami structures. Numerical examples of surfaces with modest curvature and sharp ridges are calculated to illustrate the efficiency and accuracy of our approach. A human-face case is further discussed to highlight some challenges of deployability and to demonstrate the versatility of our framework.  In the end, we expect our inverse design framework to have broad utility and be  adaptable to the many demands in engineering and architecture for functional origami structures beyond surface approximation.

\section*{Declaration of competing interest} 

The authors declare that they have no known competing financial interests or personal relationships that could have appeared to influence the work reported in this paper.

\section*{Acknowledgment} 

The authors thank Prof. Dr.~T.~Vetter, Department of Computer Science, and the University of Basel for the source of the human face data.  The authors thank the High-performance Computing Platform of Peking University for the computation resources.
This work was partly supported by the MURI program (FA9550-16-1-0566).  R.D.J. also thanks ONR (N00014-14-1-0714) and a Vannevar Bush Faculty Fellowship for partial support of this work.  X.D., H.D. and J.W. thank the National Natural Science Foundation of China (Grant
Nos.~11991033, 91848201, and 11521202) for support of this work.
{ X.D. thanks Dr. Lu Lu for assistance in laser cutting.}

\begin{appendix}
\section{Marching algorithm for deployable origami}
\label{sect:ap-marching}
We describe the marching algorithms that parameterize any origami structure $\{ \mathbf{y}^{i,j}(\boldsymbol{\alpha}_0,\boldsymbol{l}_0,\boldsymbol{\sigma}_0, \omega) | i = 0,1, \ldots, M, j = 0,1, \ldots, N\}$ obtained by folding a RFFQM.  With the exposition here, we aim for a compact description that is easy to implement numerically. We refer to our previous work \cite{feng2020designs} for a justification and detailed derivation of the formulas outlined. 

\subsection{Some preliminary definitions} Let $\theta, \varphi \in (0,\pi)$ such that $(\theta, \varphi) \neq (\pi/2, \pi/2)$.  A valid mountain-valley (M-V) assignment will be indicated by the set 
\begin{equation}
\begin{aligned}
\mathcal{MV}(\theta, \varphi) = \left\{ \begin{array}{l} -  \qquad \text{ if } \theta = \varphi \neq \pi/2 \\ +   \qquad \text{ if } \theta = \pi - \varphi \neq \pi/2  \\ \pm  \qquad  \text{ if }  \theta \neq \varphi \neq \pi - \varphi \end{array} \right\}.
\end{aligned} \label{eq:mv}
\end{equation}
For a sign $\sigma \in \mathcal{MV}(\theta, \varphi)$ indicating such an assignment, we define the folding angle functions
\begin{equation}
\begin{aligned}
&\bar{\gamma}_{\text{V}}^{\sigma}( \omega; \theta, \varphi ) = \text{sign} \Big((\sigma \cos \varphi - \cos \theta) \omega \Big) \arccos \Big( \frac{(-\sigma1 + \cos \theta \cos \varphi)\cos \omega + \sin \theta \sin \varphi }{-\sigma 1 + \cos \theta \cos \varphi + \sin \theta \sin \varphi \cos \omega } \Big), \\
&\bar{\gamma}_H^{\sigma}(\omega; \theta, \varphi) = \bar{\gamma}_V^{\sigma}(\omega; \theta, \pi - \varphi),
\end{aligned}
\end{equation}
for $\omega \in [-\pi, \pi]$ and the fold angle multipliers 
\begin{equation}
\begin{aligned}
&\mu^{\sigma}_{V} (\theta, \varphi) =  \frac{-\sigma 1 + \cos \theta \cos \varphi + \sin \theta \sin \varphi}{ \cos \varphi - \sigma \cos \theta}, \\
&\mu^{-\sigma}_{H} (\theta, \varphi) = \mu_{V}^{-\sigma} (\theta, \pi -  \varphi) .
\end{aligned}
\end{equation}
We will employ the notation $\mathbf{R}_{\mathbf{e}}(\gamma)$ for a right-hand rotation along an axis $\mathbf{e}$ (a unit vector) by an angle $\gamma$.  We will find the following matrix useful: 
\begin{equation}
\begin{aligned}
\mathbf{L}(\theta_a, \theta_b, \theta_c) =\left(\begin{array}{cc}  \tfrac{-\sin \theta_b}{\sin(\theta_a + \theta_b + \theta_c)}   &  \tfrac{\sin(\theta_a + \theta_b)}{\sin(\theta_a + \theta_b + \theta_c)} \\ 
\tfrac{\sin(\theta_a + \theta_c)}{\sin(\theta_a + \theta_b + \theta_c)} & \tfrac{-\sin \theta_c}{\sin(\theta_a + \theta_b + \theta_c)} \end{array}\right).
\end{aligned}
\end{equation}
Finally, we will always use $\mathbf{e}_{1}, \mathbf{e}_2, \mathbf{e}_3$ to denote the standard basis vectors in 3D. 

\subsection{Valid input data}   Following the notation of Fig.~\ref{fig:forward}(a), each of the sector angles on the bottom boundary must be chosen to satisfy $\alpha^{i,0}, \beta^{i,0} \in (0,\pi)$ with $(\alpha^{i,0}, \beta^{i,0}) \neq (\pi/2, \pi/2)$.  Accordingly, the M-V assignments must be chosen so that $\sigma^{i,0} \in \mathcal{MV}(\beta^{i,0}, \alpha^{i,0})$.  Finally, the lengths must be positive $w^{i,0} > 0$.  The angle and sign conditions should hold for all $i = 0, 1, \ldots, M$ and the length conditions should hold for all $i = 0,\ldots, M-1$, where $M$ is the total number of columns desired for the crease pattern.    The left boundary is similarly restricted: it is required that $\alpha^{0,j}, \beta^{0,j} \in (0,\pi)$ with $(\alpha^{0,j}, \beta^{0,j}) \neq (\pi/2, \pi/2)$ and  $\sigma^{0,j} \in \mathcal{MV}(\beta^{0,j}, \alpha^{0,j})$ for all $j = 1, \ldots, N$.  It is also required that $l^{0,j} > 0$ for all $j = 0, 1, \ldots, N-1$.  Here, $N$ is the total number of rows of the desired crease pattern.    

From hereon, we represent valid input data via an angle array $\boldsymbol{\alpha}_0$, a length array $\boldsymbol{l}_0$, and a sign array $\boldsymbol{\sigma}_0$ whose components are constrained as above.  Note, for an $M \times N$ crease pattern, the angle array has $2(M + 1) + 2N$ components, the length array has $M + N$ components, and the sign array has $M + N + 1$ components. 

\subsection{Marching to obtain the sector angles, lengths and M-V assignments}  Let $(\boldsymbol{\alpha}_0,\boldsymbol{l}_0,\boldsymbol{\sigma}_0)$ be valid input data, and suppose we have marched to the $(i,j)$-index, $i, j > 0$, without incompatibility (defined below).  Then the angles $\alpha^{i-1,j-1}$,$\beta^{i-1,j-1}$, $\alpha^{i,j-1}$,$\beta^{i,j-1}$, $\alpha^{i-1,j}$,$\beta^{i-1,j}$, M-V assignments $\sigma^{i-1,j-1}, \sigma^{i-1,j} ,\sigma^{i,j-1}$, and lengths $w^{i-1,j-1}, l^{i-1,j-1}$ are prescribed and valid.  Set
\begin{equation}
\begin{aligned}
\mu^{i,j} = \mu_H^{-\sigma^{i-1,j-1}}(\beta^{i-1,j-1} , \alpha^{i-1,j-1} ) \mu_V^{\sigma^{i-1,j}}( \pi - \alpha^{i-1,j}, \pi - \beta^{i-1,j}) \mu_V^{\sigma^{i,j-1}}( \alpha^{i,j-1}, \beta^{i,j-1}).
\end{aligned}
\end{equation}
Check the following conditions of compatibility 
\begin{equation}
\begin{aligned}
(\text{Compatibility at the (i,j)-vertex:}) \quad \begin{cases}
\beta^{i-1,j-1} + \alpha^{i,j-1} - \alpha^{i-1,j}  \in (0, \pi), \\
|\mu^{i,j}| \neq 1, \\ 
\mathbf{L}(\beta^{i-1,j-1}, \pi - \alpha^{i-1,j}, \alpha^{i,j-1})\left( \begin{array}{c} l^{i-1,j-1}  \\ w^{i-1,j-1} \end{array}\right) > \mathbf{0}.
\end{cases}
\end{aligned} \label{eq:compatibility}
\end{equation}
If all these conditions hold, set 
\begin{equation}
\begin{aligned}
&\beta^{i,j} = \beta^{i-1,j-1} + \alpha^{i,j-1} - \alpha^{i-1,j}, \quad \sigma^{i,j}1 = - \text{sign} ((\mu^{i,j})^2 - 1), \\
&\alpha^{i,j}= \arccos \Big(\sigma^{i,j} \frac{2 \mu^{i,j} + ((\mu^{i,j})^2 + 1) \cos (\beta^{i,j}) ) }{2 \mu^{i,j}  \cos (\beta^{i,j}) +  ((\mu^{i,j})^2 + 1)  } \Big), \\
&\left(\begin{array}{c} l^{i,j-1} \\ w^{i-1,j} \end{array} \right) = \mathbf{L}(\beta^{i-1,j-1}, \pi - \alpha^{i-1,j}, \alpha^{i,j-1})\left( \begin{array}{c} l^{i-1,j-1}  \\ w^{i-1,j-1} \end{array}\right).
\end{aligned} \label{eq:marching}
\end{equation}
Alternatively, if any one of the compatible conditions fails, then the input data $(\boldsymbol{\alpha}_0,\boldsymbol{l}_0,\boldsymbol{\sigma}_0)$ is not compatible with RFFQM and the marching algorithm cannot continue.

\subsection{Marching to obtain the flat crease pattern}  Assume $(\boldsymbol{\alpha}_0,\boldsymbol{l}_0,\boldsymbol{\sigma}_0)$ is compatible, so that the previous marching algorithm populated all sector angles, lengths and M-V assignments associated with the crease pattern.   We now compute the vertices associated with the flat crease pattern $\{ \mathbf{x}^{i,j}(\boldsymbol{\alpha}_0,\boldsymbol{l}_0,\boldsymbol{\sigma}_0) | i = 0,1, \ldots, M, j = 0,1, \ldots, N\}$.  For the first panel, we set
\begin{equation}
\begin{aligned}
\mathbf{x}^{0,0} = \mathbf{0}, \quad  \mathbf{x}^{1,0} = w^{0,0} \mathbf{e}_1, \quad \mathbf{x}^{0,1} =  l^{0,0} \mathbf{R}_{\mathbf{e}_3}(\beta^{0,0}) \mathbf{e}_1, \quad \mathbf{x}^{1,1} = \mathbf{x}^{1,0} - l^{1,0} \mathbf{R}_{\mathbf{e}_3}(-\alpha^{1,0}) \mathbf{e}_1. 
\end{aligned}
\end{equation} 
For the first row with $i > 1$, we set 
\begin{equation}
\begin{aligned}
\mathbf{x}^{i,0} = \mathbf{x}^{i-1,0} +\tfrac{w^{i-1,0}}{l^{i-1,0}} \mathbf{R}_{\mathbf{e}_3}(- \beta^{i-1,0}) (\mathbf{x}^{i-1,1} - \mathbf{x}^{i-1,0}), \quad \mathbf{x}^{i,1} = \mathbf{x}^{i,0} + \tfrac{l^{i,0}}{w^{i-1,0}} \mathbf{R}_{\mathbf{e}_3}(- \alpha^{i,0})(\mathbf{x}^{i-1,0} -  \mathbf{x}^{i,0}).
\end{aligned}
\end{equation}
For everything else, i.e., $j > 1$, we set
\begin{equation}
\begin{aligned}
\mathbf{x}^{i,j} = \mathbf{x}^{i,j-1} +\tfrac{l^{i,j-1}}{l^{i,j-2}} \mathbf{R}_{\mathbf{e}_3}(\pi - \alpha^{i,j-1} + \beta^{i,j-1}) (\mathbf{x}^{i,j-2} - \mathbf{x}^{i,j-1}).
\end{aligned}
\end{equation}

\subsection{Marching to obtain the folding angles} Assume $(\boldsymbol{\alpha}_0,\boldsymbol{l}_0,\boldsymbol{\sigma}_0)$ is compatible and take the quantities in the marching algorithms in C-D as given.  Fix a folding parameter $\omega \in [0,\pi]$.  We define the folding angles at horizontal and vertical creases, respectively, as 
\begin{equation}
\begin{aligned}
\gamma_{H}^{i,j} = \begin{cases}
-\sigma^{i,j} \gamma_{H}^{i-1,j} & \text{ if } i \neq 0,\\
\sigma^{i,j} \bar{\gamma}_{H}^{-\sigma^{i,j}} (\gamma_V^{i,j-1}; \beta^{i,j}, \alpha^{i,j}) & \text{ if } i = 0, j \neq 0, \\
\sigma^{i,j} \bar{\gamma}_{H}^{-\sigma^{i,j}} (\omega; \beta^{i,j}, \alpha^{i,j}) & \text{ if } i = j = 0, 
\end{cases} \qquad \gamma_{V}^{i,j} = \bar{\gamma}_V^{\sigma^{i,j}}(\gamma_H^{i,j}; \beta^{i,j}, \alpha^{i,j}).
\end{aligned} \label{eq:folding_angle}
\end{equation}

\subsection{Marching to obtain the deformation gradients} Assume $(\boldsymbol{\alpha}_0,\boldsymbol{l}_0,\boldsymbol{\sigma}_0)$ is compatible and take the quantities in the marching algorithms in C-E as given.  Assume $i < M, j < N$.  Set $\mathbf{t}^{i,j} = \tfrac{\mathbf{x}^{i,j+1} - \mathbf{x}^{i,j}}{|\mathbf{x}^{i,j+1} - \mathbf{x}^{i,j}|}$ and $\mathbf{s}^{i,j} = \tfrac{\mathbf{x}^{i+1,j} - \mathbf{x}^{i,j}}{|\mathbf{x}^{i+1,j} - \mathbf{x}^{i,j}|}$.  Then  set
\begin{equation}
\begin{aligned}
\mathbf{F}^{i,j}  = \begin{cases}
\mathbf{F}^{i-1,j} \mathbf{R}_{\mathbf{t}^{i,j}}(\gamma_{V}^{i,j}) & \text{ if }  i \neq 0, \\
\mathbf{F}^{i,j-1} \mathbf{R}_{\mathbf{s}^{i,j}}(-\gamma_{H}^{i,j}) & \text{ if } i =0 , j \neq 0, \\
\mathbf{I} & \text{ if } i = j = 0.
\end{cases}
\end{aligned} \label{eq:deformation_gradient}
\end{equation}

\subsection{Marching to obtain the origami structure} Assume $(\boldsymbol{\alpha}_0,\boldsymbol{l}_0,\boldsymbol{\sigma}_0)$ is compatible and take the quantities in the marching algorithms in C-F as given. We finally compute the vertices associated with the origami $\{ \mathbf{y}^{i,j}(\boldsymbol{\alpha}_0,\boldsymbol{l}_0,\boldsymbol{\sigma}_0, \omega) | i = 0,1, \ldots, M, j = 0,1, \ldots, N\}$.  These vertices are given by 
\begin{equation}
\begin{aligned}
\mathbf{y}^{i,j} =\begin{cases}
\mathbf{y}^{i-1,j} + \mathbf{F}^{i-1,j}(\mathbf{x}^{i,j} - \mathbf{x}^{i-1,j})  & \text{ if } i \neq 0, j < N, \\
\mathbf{y}^{i-1,j} +   \mathbf{F}^{i-1,j-1}(\mathbf{x}^{i,j} - \mathbf{x}^{i-1,j})  & \text{ if } i \neq 0, j = N, \\
\mathbf{y}^{i,j-1} + \mathbf{F}^{i,j-1}(\mathbf{x}^{i,j} - \mathbf{x}^{i,j-1}) & \text{ if } i = 0, j \neq 0, \\
\mathbf{x}^{i,j} & \text{ if } i = j = 0. 
\end{cases}
\end{aligned} \label{eq:deformation_function}
\end{equation}

\subsection{Lemma for the compatible input data}
\begin{lemma}
Let $(\bar{\boldsymbol{\alpha}}_0, \bar{\boldsymbol{l}}_0, \bar{\boldsymbol{\sigma}}_0)$ be compatible input data.   Then the following two results hold when fixing the M-V assignment indicated by the array of signs $\bar{\boldsymbol{\sigma}}_0$:
\begin{itemize}[leftmargin=*]
\item  There is an open neighborhood $\mathcal{N}(\bar{\boldsymbol{\alpha}}_0, \bar{\boldsymbol{l}}_0) \subset \mathbb{R}^{2(M+N+1)} \times \mathbb{R}^{M+N}$ of the point $(\bar{\boldsymbol{\alpha}}_0, \bar{\boldsymbol{l}}_0)$ for which $(\boldsymbol{\alpha}_0, \boldsymbol{l}_0, \bar{\boldsymbol{\sigma}}_0)$ is compatible input data for all $(\boldsymbol{\alpha}_0, \boldsymbol{l}_0) \in \mathcal{N}(\bar{\boldsymbol{\alpha}}_0, \bar{\boldsymbol{l}}_0)$.
\item The parameterization $\mathbf{y}^{i,j}_{\boldsymbol{\sigma}^{\star}_0} (\boldsymbol{\alpha}_0, \boldsymbol{l}_0, \omega) =  \mathbf{y}^{i,j}(\boldsymbol{\alpha}_0, \boldsymbol{l}_0, \bar{\boldsymbol{\sigma}}_0, \omega)$  is  smooth on $\mathcal{N}(\bar{\boldsymbol{\alpha}}_0, \bar{\boldsymbol{l}}_0)  \times (0,\pi)$ for every $(i,j)$-vertex of the $M \times N$ origami pattern.
\end{itemize}
\end{lemma}
\begin{proof}
For the fixed M-V assignment $\bar{\boldsymbol{\sigma}}_0$, the components of the given compatible input data $\bar{\boldsymbol{\alpha}}_0$ and $\bar{\boldsymbol{l}}_0$ are the boundary sector angle pairs $(\bar{\alpha}^{0,j}, \bar{\beta}^{0,j})$, $(\bar{\alpha}^{i,0}, \bar{\beta}^{i,0})$ and the boundary side lengths $\{\bar{w}^{0,j}, \bar{l}^{0,j}, \bar{w}^{i,0}, \bar{l}^{i,0}\}$, which are in the open sets $((0, \pi)\times(0,\pi) \setminus (\pi/2, \pi/2))$ and $(0, +\infty)$ respectively. To prove the existence of compatible $\mathcal{N}(\bar{\boldsymbol{\alpha}}_0, \bar{\boldsymbol{l}}_0)$, we firstly notice that the M-V assignment $\bar{\boldsymbol{\sigma}}_0$ is valid in a small enough open neighborhood of $\bar{\boldsymbol{\alpha}}_0$ by satisfying the validity condition Eq.~(\ref{eq:mv}), because the marching function for $\sigma^{i,j}$ in Eq.~(\ref{eq:marching}) is continuous when the given pattern generated by $(\bar{\boldsymbol{\alpha}}_0, \bar{\boldsymbol{l}}_0)$ is compatible. We also notice that the marching functions for the sector angles $\alpha^{i,j}, \beta^{i,j}$ and the side lengths $w^{i,j}, l^{i,j}$ are  continuous functions which map  open sets to open sets. If the  open set $\mathcal{N}(\bar{\boldsymbol{\alpha}}_0, \bar{\boldsymbol{l}}_0)$ is small enough, the marching functions will ensure the validity of the vertex compatibility Eq.~(\ref{eq:compatibility}) according to the compatibility's openness, and therefore ensure that the  sector angles and the side lengths over the pattern  are all compatible. Thus, there exists a small enough open neighborhood  $\mathcal{N}(\bar{\boldsymbol{\alpha}}_0, \bar{\boldsymbol{l}}_0)$ of the point $(\bar{\boldsymbol{\alpha}}_0, \bar{\boldsymbol{l}}_0)$ such that all the points in $\mathcal{N}(\bar{\boldsymbol{\alpha}}_0, \bar{\boldsymbol{l}}_0)$ result in valid RFFQM patterns. Next, we prove the vertex position function $\mathbf{y}^{i,j}(\boldsymbol{\alpha}_0, \boldsymbol{l}_0, \bar{\boldsymbol{\sigma}}_0, \omega)$ is smooth for $(\boldsymbol{\alpha}_0, \boldsymbol{l}_0) \in \mathcal{N}(\bar{\boldsymbol{\alpha}}_0, \bar{\boldsymbol{l}}_0)$ and for $\omega \in (0, \pi)$. For the $(i,j)$ vertex, the deformation Eq.~(\ref{eq:deformation_function}), calculated from the folding angle function Eq.~(\ref{eq:folding_angle}),  the deformation gradient Eq.~(\ref{eq:deformation_gradient}) and the previous vertex positions, is smooth for $\omega \in (0, \pi)$ and $(\boldsymbol{\alpha}_0, \boldsymbol{l}_0) \in \mathcal{N}(\bar{\boldsymbol{\alpha}}_0, \bar{\boldsymbol{l}}_0)$ , because  Eq.~(\ref{eq:folding_angle}) and  Eq.~(\ref{eq:deformation_gradient}) are valid and smooth for compatible input data $(\boldsymbol{\alpha}_0, \boldsymbol{l}_0)$ (and its small enough open neighborhood). This completes the proof.
\end{proof}

\section{Details of the optimization procedure}
\label{sect:ap-procedure}
\subsection{Procedure for discretization} A partially-folded Miura-Ori is given by input data $(\boldsymbol{\alpha}_{\text{M-O}}, \boldsymbol{l}_{\text{M-O}}, \boldsymbol{\sigma}_{\text{M-O}})$ in Eq.~(\ref{eq:inputData}) and a folding parameter  $0< \omega_{\text{M-O}} < \pi$.  For notational convenience, we assume that $M$ (the number of columns of panels) and $N$ (the number of rows) of the origami are both even.  Then, the collection of vertices with this input data (Eq.~(\ref{eq:rMO})) provides a discretization of a planar region corresponding to two identical rectangular lattices shifted by an offset.  (Here, $ M/2$ is even by assumption.) This offset lattice has a rectangular unit cell of side length $W$ parallel to a unit vector $\mathbf{r}_u$ and side length $L$ parallel to a unit vector $\mathbf{r}_v$ (orthogonal to $\mathbf{r}_u$). It also has an identical rectangular unit cell shifted by a vector  $\delta_u W \mathbf{r}_u + \delta_v L \mathbf{r}_v$. These quantities are all explicit functions of the vertices defined above, i.e., 
\begin{equation}
\begin{aligned}\label{eq:LWDefine}
&W = |\mathbf{r}_{\text{M-O}}^{1,0} - \mathbf{r}_{\text{M-O}}^{0,0}|, && L = |\mathbf{r}_{\text{M-O}}^{0,2} - \mathbf{r}_{\text{M-O}}^{0,0}|, \\
&\mathbf{r}_u = \frac{1}{W}( \mathbf{r}_{\text{M-O}}^{1,0} - \mathbf{r}_{\text{M-O}}^{0,0}), && \mathbf{r}_v = \frac{1}{L} (\mathbf{r}_{\text{M-O}}^{0,2} - \mathbf{r}_{\text{M-O}}^{0,0}) , \\
&\delta_u =  \frac{\mathbf{r}_{\text{M-O}}^{0,1} - \mathbf{r}_{\text{M-O}}^{0,0}}{|\mathbf{r}_{\text{M-O}}^{0,1} - \mathbf{r}_{\text{M-O}}^{0,0}|} \cdot \mathbf{r}_u, &&  \delta_v  =  \frac{\mathbf{r}_{\text{M-O}}^{0,1} - \mathbf{r}_{\text{M-O}}^{0,0}}{|\mathbf{r}_{\text{M-O}}^{0,1} - \mathbf{r}_{\text{M-O}}^{0,0}|} \cdot \mathbf{r}_v.
\end{aligned}
\end{equation}

\begin{figure}[t!]
\centering
\includegraphics[width=0.9\linewidth]{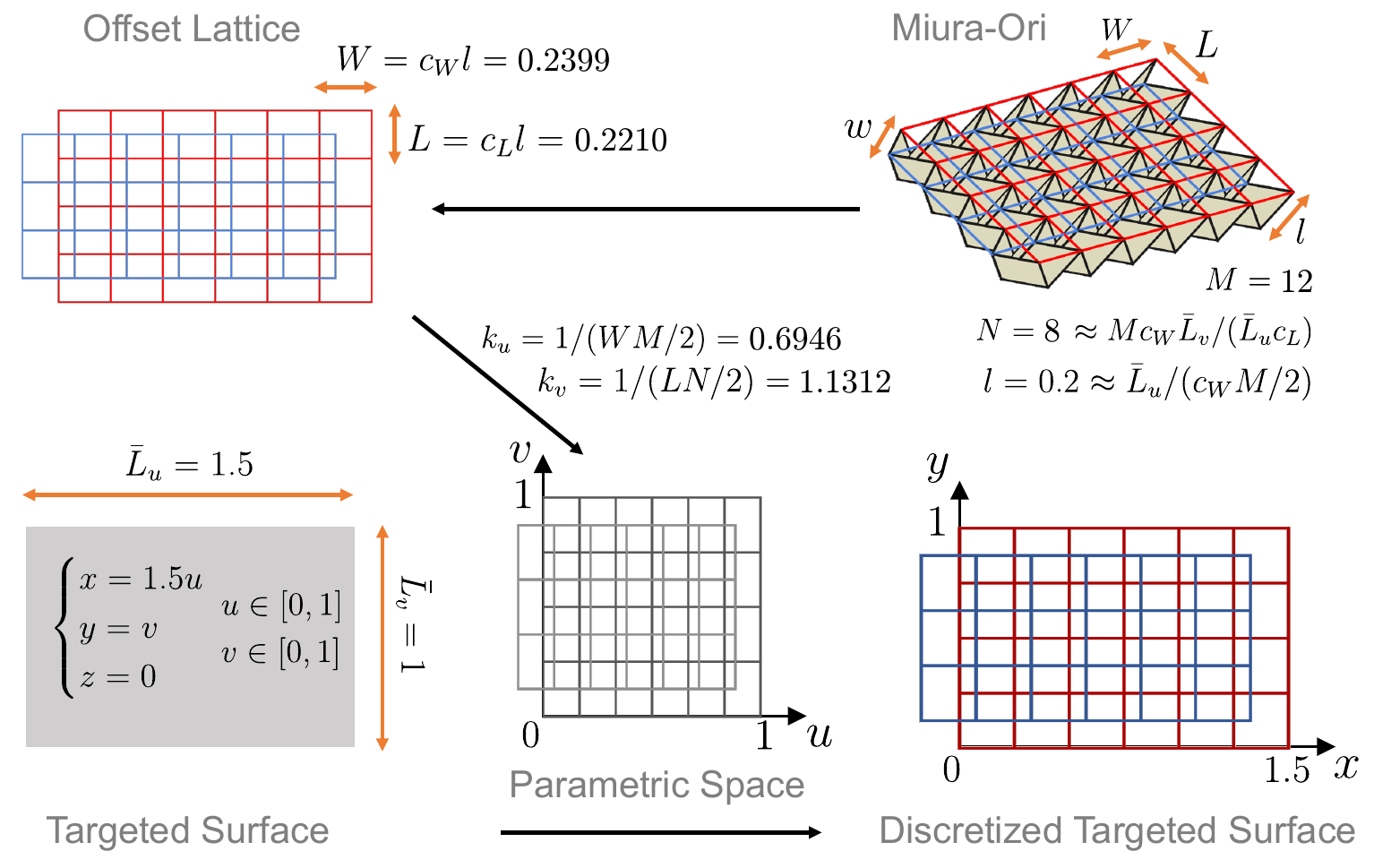}	
\caption{Projection of the Miura-Ori offset lattice onto the parametric space of a planar targeted surface. The parameters of the Miura-Ori $l, w, L, W, M, N, \alpha, \omega_{\text{M-O}}$ are chosen to match the size $(\bar{L}_u, \bar{L}_v)$ of the targeted surface. The offset lattice is obtained by projecting the Miura-Ori onto the 2D plane. By the rescaling factors $k_u$ and $k_v$, the parametric space is matched to the offset lattice, and then discretized accordingly. See text for details. } 
\label{fig:ap-scaling}
\end{figure}

We discretize the targeted surface so as to match an offset lattice associated to a partially-folded Miura-Ori. For the targeted surface, we consider an arbitrary regular parametric surface $\bar{\bfr}(u,v)$  of a rectangular domain:
\beq\label{eq:targetParam}
\bar{\bfr}(u,v)=(x(u,v),y(u,v),z(u,v)), \quad u \in [u^L, u^R],  \quad v \in [v^L, v^R]
\eeq
where $u^L < u^R$ and $v^L < v^R$. For the discretization based on the offset lattice, we treat the general case here (though we use a specific offset lattice in calculations for simplicity). In this setting, we prescribe $w = a_r l$  for an aspect ratio $a_r > 0$.  As a result, every side-length of the Miura-Ori is proportional to $l$ in this prescription.  Elasticity scaling therefore dictates that the side lengths of the offset lattice in Eq.~(\ref{eq:LWDefine}) satisfy $W = c_W l$ and $L = c_L l$ for positive numbers $c_W = c_W(\alpha, \omega_{\text{M-O}}, a_r)$ and $c_L = c_L(\alpha, \omega_{\text{M-O}}, a_r)$. We  henceforth treat the quantities $\alpha, \omega_{\text{M-O}}$ and  $a_r$ as given (design parameters one can freely choose).   Since the Miura-Ori has $M$ columns and $N$ rows of panels (both even numbers), the red lattice on the Miura-Ori (Fig.~\ref{fig:ap-scaling}) has a width $WM/2$ and length $LN/2$. As we already know the equations of the targeted surface, the characteristic lengths $\bar{L}_u$ and $\bar{L}_v$ can be estimated to represent the total size of the targeted surface along the $u$ and $v$ directions, respectively. We therefore choose the parameters $M, N, l$ such that $c_W l M/2   \approx \bar{L}_u$ and $c_L l  N/2  \approx \bar{L}_v$.   Precisely we fix an even integer $M$ to dictate the total number of panels desired for the optimization.  Then we choose an even integer $N$ that best approximates the ratio $\frac{c_W M/2 }{c_LN/2 } \approx \frac{\bar{L}_u}{\bar{L}_v}$. Finally, we set $l$ such that $c_W l  M/2   \approx \bar{L}_u$.

With all the parameters set, we project the Miura-Ori offset lattice in Eq.~(\ref{eq:rMO}) onto a plane by the formulas
\begin{equation}
\begin{aligned}
\left[\begin{array}{c} u^{i,j} \\  v^{i,j} \end{array}\right] =\left[ \begin{array}{c} k_u\mathbf{r}_{\text{M-O}}^{i,j} \cdot \mathbf{r}_{u} + c_u \\  k_v\mathbf{r}_{\text{M-O}}^{i,j} \cdot \mathbf{r}_{v} + c_v \end{array}\right],
\end{aligned}
\end{equation}
choosing the scaling $(k_u,k_v)$ and translation $(c_u, c_v)$ so that the center of the rectangular region $[u^L, u^R] \times [v^L, v^R]$ and the center of the projected offset lattice match.  In other words, the average lattices $\langle (u^{i.j}, v^{i,j}) \rangle$ should coincide with the center of the rectangular region.  We therefore obtain the formulas for this transformation
\begin{equation}
\begin{aligned}
\left[\begin{array}{c} k_u \\ k_v \end{array}\right]= \left[\begin{array}{c} (u^R-u^L)/(W M/2) \\ (v^R-v^L)/(L N/2) \end{array}\right],
\end{aligned}
\end{equation}

\begin{equation}
\begin{aligned}
\left[\begin{array}{c} c_u \\ c_v \end{array}\right]=  \left[\begin{array}{c} u^L - k_u\mathbf{r}_{\text{M-O}}^{0,0} \cdot \mathbf{r}_{u} \\ v^L - k_v\mathbf{r}_{\text{M-O}}^{0,0} \cdot \mathbf{r}_{v} \end{array}\right].
\end{aligned}
\end{equation}
The projection process of a targeted plane is sketched in Fig.~\ref{fig:ap-scaling}. Note, by following the procedure above exactly, some boundary points such as $(u^{0,2j}, v^{0,2j})$ would exceed the given domain $[u^L, u^R] \times [v^L, v^R]$ due the offset. This is not an issue because either the domain of parametric equation can be enlarged slightly, or we can decrease slightly the scaling coefficients $k_u$ and $k_v$ to include all the points in the given domain.  The latter can always be done. 

\subsection{Some remarks on the optimization procedure}

\begin{figure}[t!]
\centering
\includegraphics[width=0.9\linewidth]{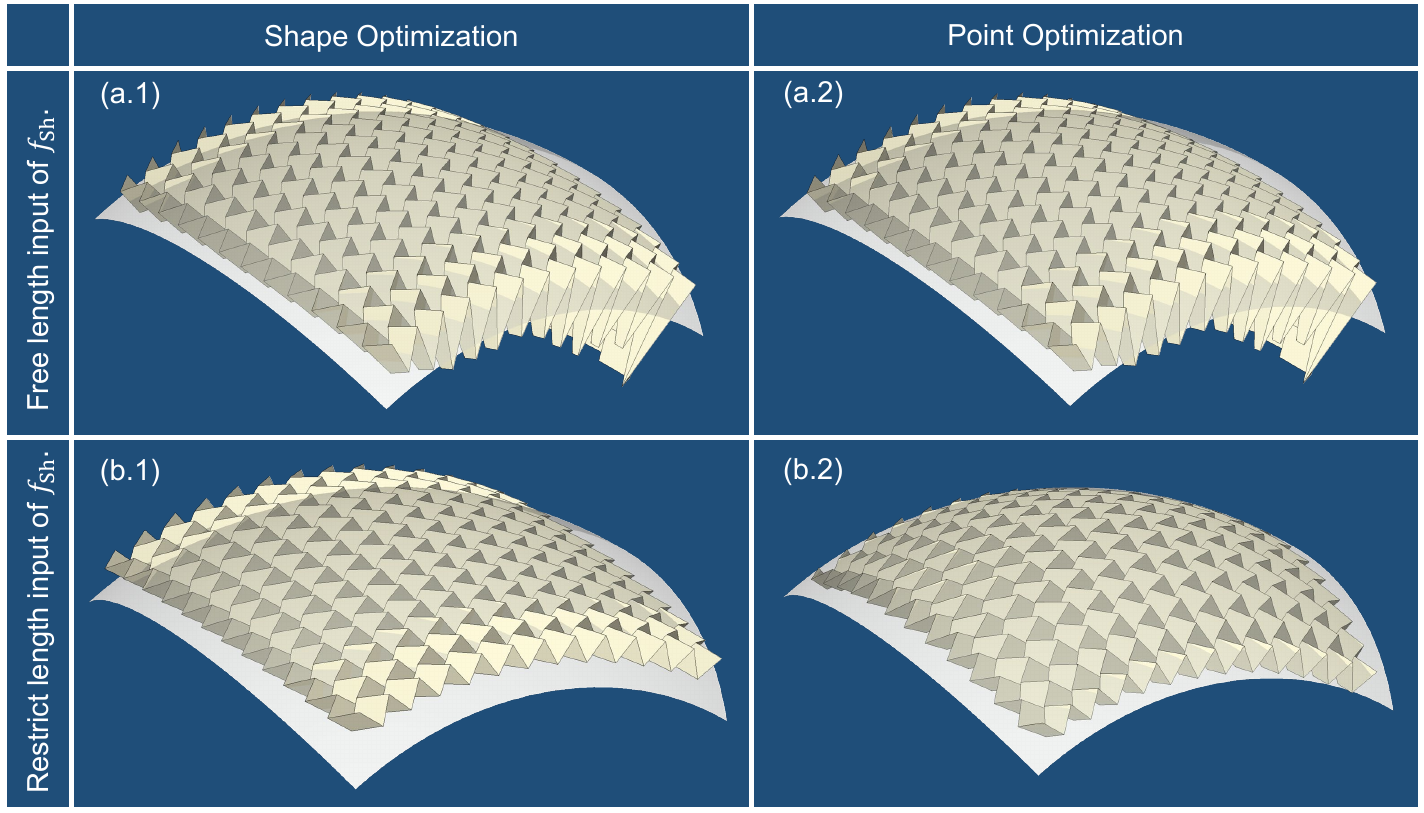}	
\caption{Results of spherical cap approximation: (a) When the boundary length input data of $f_{\rm{Sh.}}$ is optimized freely, this leads to distorted aspect ratios. In addition, point optimization exits quickly with negligible refinements. (b) When the boundary length input data of $f_{\rm{Sh.}}$ is restricted, i.e., $\boldsymbol{l}_0 = l_0 \mathbf{1}$, this leads to regular aspect ratios and satisfying refinements. Note that the boundary length input data of $f_{\rm{Pt.}}$ is optimized freely in both (a) and (b).}
\label{fig:ap-distort}
\end{figure}

We develop a two-stage optimization procedure that produces two deployable origami structures,
\begin{equation}
\begin{aligned}
\text{after shape optimization:} \quad &\big\{ \mathbf{y}^{i,j}(\boldsymbol{\alpha}_0^{\star}, l_0^{\star}  \boldsymbol{1}, \boldsymbol{\sigma}_{\text{M-O}}, \omega^{\star})  \big| i = 0,1,\ldots, M, j= 0,1,\ldots, N\big\}, \\
\text{after point optimization:} \quad &\big\{ \mathbf{y}^{i,j}(\boldsymbol{\alpha}_0^{\star \star}, \boldsymbol{l}_0^{\star \star}, \boldsymbol{\sigma}_{\text{M-O}}, \omega^{\star \star})  \big| i = 0,1,\ldots, M, j= 0,1,\ldots, N\big\}.
\end{aligned}
\end{equation}
When optimizing for shape to obtain the first deployable origami structure, we restrict the boundary length input data $l_0 \mathbf{1}$, only allowing for a rescaling of the Miura-Ori boundary lengths $l \mathbf{1}$.  This is a choice based on {trial and error} in our numerical investigation.  Because we only control the top vertices of the origami structure, additional freedom to vary the boundary lengths can lead to distorted aspect ratios and origami that is not conducive to manufacturability or deployability considerations.  An example to this effect is highlighted in Fig.~\ref{fig:ap-distort}.

Note, unlike shape optimization, we allow the boundary lengths $\boldsymbol{l}_0$ to freely vary in point optimization.  This additional freedom does not generally pose the same kind of aspect ratio and  manufacturability issues observed in the shape optimization.  A basic heuristic for why is that the origami surface optimized for shape is already a decent candidate for point optimality.  Thus, $|(\boldsymbol{\alpha}_0^{\star \star}, \boldsymbol{l}_0^{\star \star}, \omega^{\star \star}) -  ( \boldsymbol{\alpha}_0^{\star},  l_0^{\star} \boldsymbol{1}, \omega^{\star})|$ is typically small. 

Finally, in the calculations of shape operators and in measuring the quality of approximation, we have non-dimensionalized length quantities by $\langle L \rangle$ ``the average length of the quad-mesh edges of the targeted surface".  To be clear, these quads have vertices  $\{ \bar{\mathbf{r}}^{0,0}, \bar{\mathbf{r}}^{1,0}, \bar{\mathbf{r}}^{1,2}, \bar{\mathbf{r}}^{0,2}\}$,  $\{ \bar{\mathbf{r}}^{0,1}, \bar{\mathbf{r}}^{1,1}, \bar{\mathbf{r}}^{1,3}, \bar{\mathbf{r}}^{0,3} \}$, $\ldots$, etc.  So  $\langle L \rangle$ averages the side lengths of all quads defined in this fashion.

\section{Observations for input origami to the optimization}
\label{sect:ap-observations}
\subsection{Origami patterns for approximating basic surfaces}
To approximate the target surface accurately, it is efficient to choose the initial input origami close to the target surface. Here we provide some experimental results for approximating four types of basic surfaces: planar, vertical bending, horizontal bending and twisting surfaces, as the guidance for selecting the initial input.
We follow the notations in \ref{sect:ap-marching} and apply the marching algorithm to generate $2M\times 2N$ RFFQM patterns from the boundary data. The folding topology on the boundary here is the same as Miura-Ori, i.e., $\boldsymbol{\sigma} = \{ +, \ldots, +\}$, while the effect of changing folding topologies is discussed later. 
We use two parameters $\alpha\in(0,\pi)$ and $\beta\in(0,\pi)$ to represent the input sector angles, while excluding the degenerate case $\alpha=\beta=\pi/2$ that does not follow the kinematics in the framework of the marching algorithm.

\begin{figure}[!t]
\centering
\includegraphics[width=0.9\linewidth]{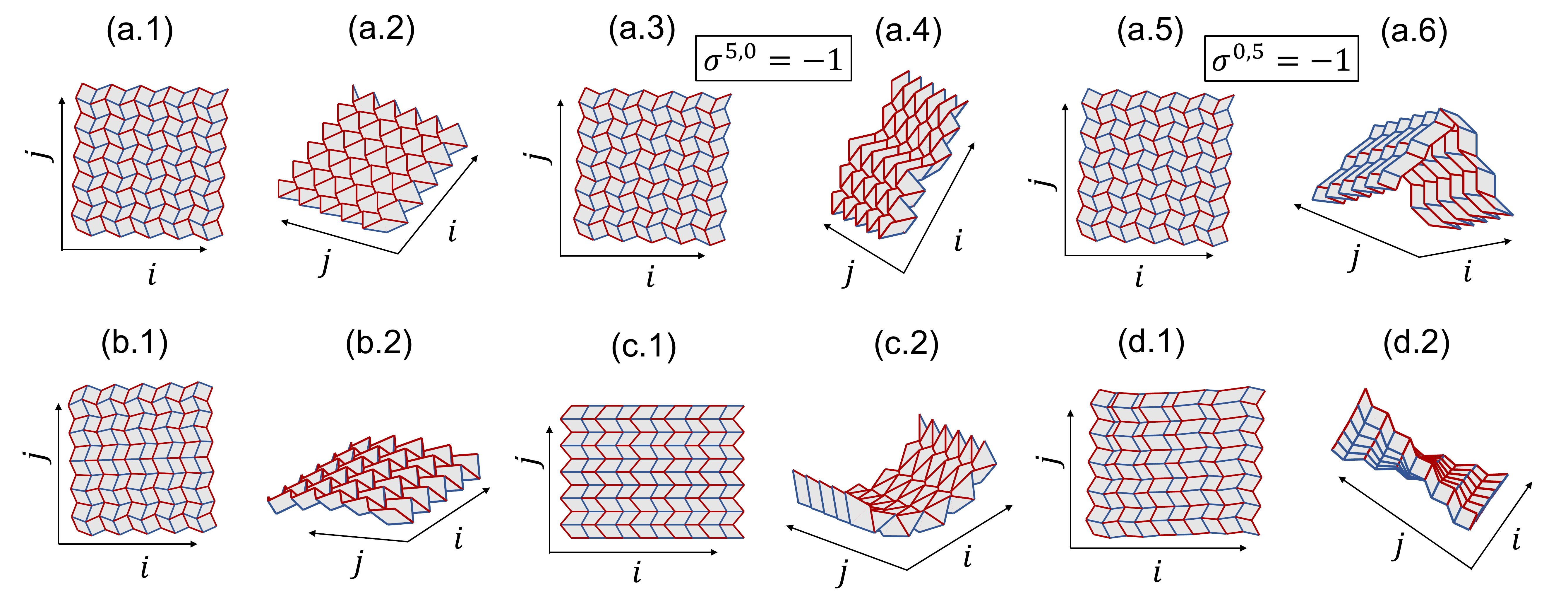}	
\caption{{ The crease patterns and folded configurations ($\omega=3\pi/4$) of origami for approximating basic deformations. The red and blue solid lines indicate the M-V assignments.} ({a}1-{a}2) Planar origami $P^{\rm{pl}}_{10,10}(\pi/2,3\pi/4)$ and (a3-a6) sharp interfaces emerging from changing of folding topologies. (b) Vertical bending origami $P^{\rm{vb}}_{10,10}(\pi/2,3\pi/4)$. (c) Horizontal bending origami $P^{\rm{hb}}_{10,10}(11\pi/36,25\pi/36,\pi/18)$. (d) Twisting origami $P^{\rm{ts}}_{10,10}(7\pi/18,13\pi/18)$}. 
\label{fig:ap-deformation}
\end{figure}

\begin{enumerate}
\item Planar origami $P_{2M,2N}^{\rm{pl}}(\alpha,\beta)$: 
\begin{equation}
\begin{aligned} \label{eq:pl}
&\alpha^{0,0}= \alpha^{2,0} = \ldots = \alpha, && \beta^{0,0} = \beta^{2,0}  = \ldots =\beta, \\
&\alpha^{1,0}=  \alpha^{3,0} = \ldots =  \pi -\beta, && \beta^{1,0}= \beta^{3,0} =\ldots= \pi -\alpha,\\
&\alpha^{0,2}= \alpha^{0,4} =\ldots = \alpha, && \beta^{0,2}= \beta^{0,4} = \ldots = \beta,\\
&\alpha^{0,1}= \alpha^{0,3} = \ldots = \beta, && \beta^{0,1} = \beta^{0,3} = \ldots =\alpha.
\end{aligned}
\end{equation}
Note, the surface remains planar during the folding for this case  (see Figs.~\ref{fig:ap-deformation}{(a.1) and (a.2)}).  Note also, these planar origami degenerate to the Miura-Ori $P_{2M,2N}^{\rm{mu}}(\alpha,\beta)$ when $\alpha+\beta=\pi$.

\item Vertical bending origami $P_{2M,2N}^{\rm{vb}}(\alpha,\beta)$:
\begin{equation}
\begin{aligned} \label{eq:vb}
&\alpha^{0,0}= \alpha^{2,0} = \ldots = \alpha, && \beta^{0,0} = \beta^{2,0}  = \ldots =\beta, \\
&\alpha^{1,0}=  \alpha^{3,0} = \ldots =  \pi -\beta, && \beta^{1,0}= \beta^{3,0} =\ldots= \pi -\alpha,\\
&\alpha^{0,2j}=\frac{N-j}{N}\alpha+\frac{j}{N}\left(\pi-\beta\right), && \beta^{0,2j}=\frac{N-j}{N}\beta+\frac{j}{N}\left(\pi -\alpha\right),\\
&\alpha^{0,2j-1}=\frac{N-j}{N-1}\beta+\frac{j-1}{N-1}\left(\pi -\alpha\right), && \beta^{0,2j - 1}=\frac{N-j}{N-1}\alpha+\frac{j-1}{N-1}\left(\pi -\beta\right).
\end{aligned}
\end{equation}
for $j = 1,\ldots, N$. The top and bottom surfaces will bend in $j$ direction but keep $i$ direction straight during the folding for this case. (see Fig.~\ref{fig:ap-deformation}{(b)}).

\item Horizontal bending origami $P_{2M,2N}^{\rm{hb}}(\alpha,\beta,\gamma)$:
\begin{equation}
\begin{aligned}  \label{eq:hb}
&\alpha^{0,0}= \alpha^{2,0} = \ldots = \alpha, && \beta^{0,0} = \beta^{2,0}  = \ldots =\beta, \\
&\alpha^{1,0}=  \alpha^{3,0} = \ldots = \alpha + \gamma, && \beta^{1,0}= \beta^{3,0} =\ldots= \beta - \gamma,\\
&\alpha^{0,2}= \alpha^{0,4} =\ldots = \alpha, && \beta^{0,2}= \beta^{0,4} = \ldots = \beta,\\
&\alpha^{0,1}= \alpha^{0,3} = \ldots = \beta, && \beta^{0,1} = \beta^{0,3} = \ldots =\alpha.
\end{aligned}
\end{equation}
for $\alpha+\beta= \pi$, $\gamma\in(0,\beta)$.  The surface will bend in $i$ direction but keep $j$ direction straight during the folding for this case (see Fig.~\ref{fig:ap-deformation}{(c)}) .  This observation can also be found in  \cite{Wang2016Folding}.

\item Twisting origami $P_{2M,2N}^{\rm{ts}}(\alpha,\beta)$:
\begin{equation}
\begin{aligned} \label{eq:tw}
&\alpha^{2i,0}=\frac{M-i}{M}\alpha+\frac{i}{M}\left(\pi -\beta\right), && \beta^{2i,0}=\frac{M-i}{M}\beta+\frac{i}{M}\left(\pi -\alpha\right),\\
&\alpha^{2i-1,0}=\frac{M-i}{M-1}(\pi-\beta)+\frac{i-1}{M-1}\alpha, && \beta^{2i-1,0}=\frac{M-i}{M-1}(\pi -\alpha)+\frac{i-1}{N-1}\beta,\\
&\alpha^{0,0} = \alpha^{0,2}= \alpha^{0,4} =\ldots = \alpha, && \beta^{0,0} =  \beta^{0,2}= \beta^{0,4} = \ldots = \beta,\\
&\alpha^{0,1}= \alpha^{0,3} = \ldots = \beta, && \beta^{0,1} = \beta^{0,3} = \ldots =\alpha.
\end{aligned}
\end{equation}
for $i = 1,\ldots, M$. The surface will exhibit a twisting motion during the folding for this case (see Fig.~\ref{fig:ap-deformation}(d)) .
\end{enumerate}

These basic origami patterns are perturbations of the Miura-Ori that can help select initial inputs intuitively for approximating different target surfaces. However, we observe experimentally that the Miura-Ori as an input is adequate for approximating target surfaces with slightly changing curvatures. 
In our results, we have accurate approximations of most surfaces while taking the Miura-Ori as the initial origami, except for the human face case, in which we use the $P^{\rm{vb}}$ origami, and the sharp-interface cases, in which we use the $P^{\rm{pl}}$ origami.

\begin{figure}[t!]
\centering
\includegraphics[width=0.9\linewidth]{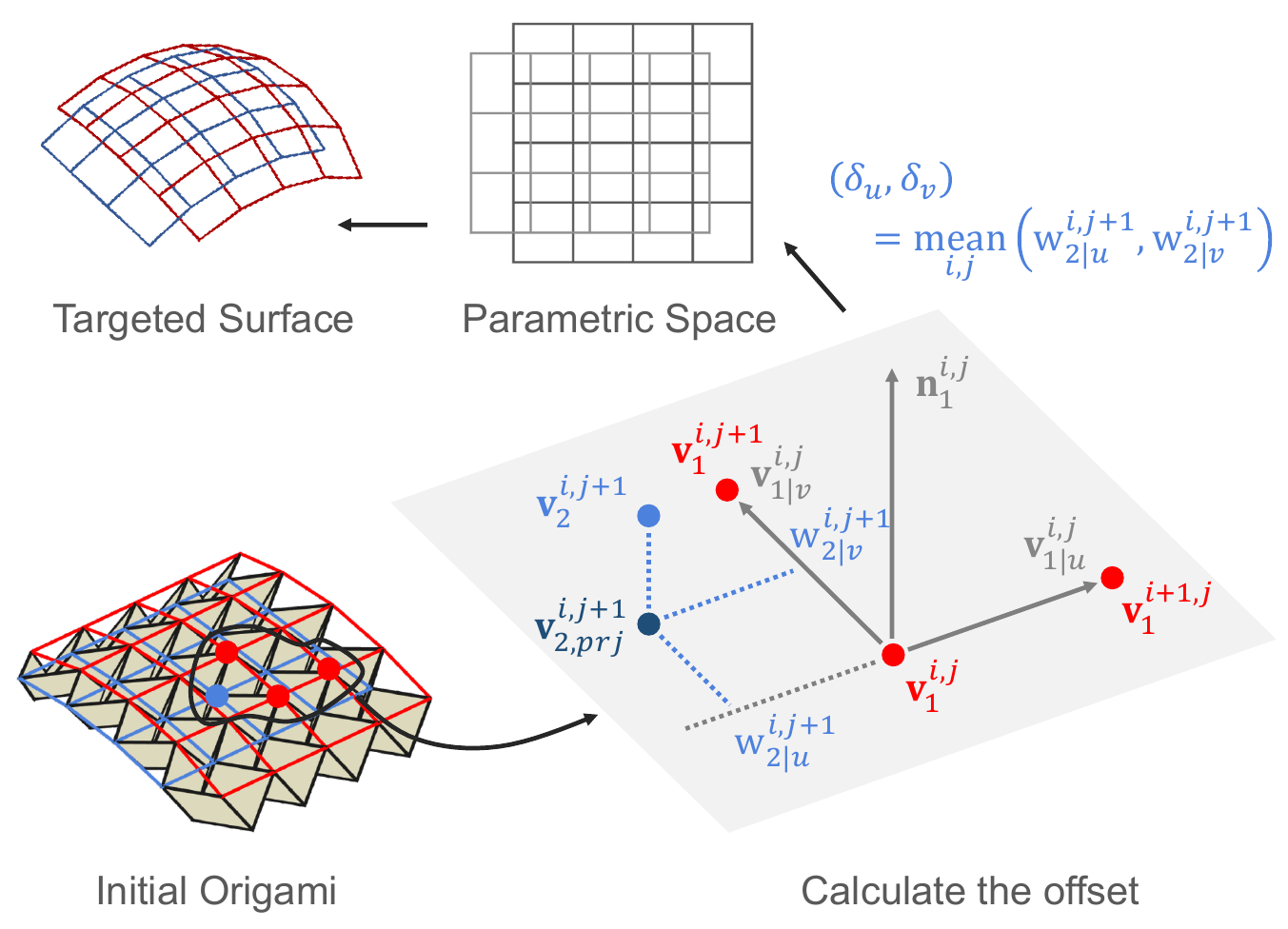}	
\caption{Discretization of the parametric surface. The offset $(\delta_u,\delta_v)$ is selected by the average of local coordinates of $\bfv_{2, prj}^{i, j+1}$, which is the projected point of $\bfv_2^{i,j+1}$ onto the $(\bfv_{1|u}^{i,j}, \bfv_{1|v}^{i,j+1})$ plane. }
\label{fig:ap-offset}
\end{figure}

\subsection{Selecting the offset} If the input origami to the optimization is a curved surface, we still require a discretization of the targeted surface.  In these cases, we again employ an offset rectangular lattice to discretize the surface ---  but with a reasoned comparison to the origami surface by averaging.  Here, we describe an approach to determine the offset, as illustrated in Fig.~\ref{fig:ap-offset}.
We employ the local basis for the red points on the origami surface:
\begin{align}
&\bfv_{1|u}^{i,j}=\bfv_1^{i+1,j}-\bfv_1^{i,j}, \quad \bfv_{1|v}^{i,j}=\bfv_1^{i,j+1}-\bfv_1^{i,j},\quad \bfv_{1|v}^{i,N}=\bfv_{1|v}^{i,N-1}, \quad \bfn_1^{i,j}=\frac{\bfv_{1|u}^{i,j}\times\bfv_{1|v}^{i,j}}{|\bfv_{1|u}^{i,j}\times\bfv_{1|v}^{i,j}|}.
\end{align}
Now consider the adjacent blue vertex $\bfv_2^{i,j+1}$ shown. This vertex is displaced from $\mathbf{v}_1^{i,j}$ by a vector $w_{2|u}^{i,j+1} \bfv_{1|u}^{i,j} + w_{2|v}^{i,j+1}  \bfv_{1|v}^{i,j} +w_{2|n}^{i,j+1}  \bfn_1^{i,j} $ with the components given by 
\beq
(w_{2|u}^{i,j+1},w_{2|v}^{i,j+1},w_{2|n}^{i,j+1}):=[{\bfv}_{1|u}^{i,j}, {\bfv}_{1|v}^{i,j}, {\bfn_1^{i,j}}]^{-1} \ (\bfv_2^{i,j+1}-\bfv^{i,j}_1).
\eeq
Note, the normal component $w_{2|n}^{i,j+1}$ vanishes for a Miura origami. To some extent, $(w_{2|u}^{i,j+1},w_{2|v}^{i,j+1})$ represents the misfit between two quad meshes along the tangent directions. So we take the average local coordinates $\mathop{\rm{mean}}\limits_{i,j}(w_{2|u}^{i,j+1},w_{2|v}^{i,j+1})$ as the offset, denoted by $(\delta_u,\delta_v)$.  We also choose the width of our rectangular unit cell as $W = \mathop{\rm{mean}}\limits_{i,j} |\mathbf{v}_{1|u}^{i,j}|$ and the length as $L = \mathop{\rm{mean}}\limits_{i,j} |\mathbf{v}_{1|v}^{i,j}|$. Given $W,L,\delta_u, \delta_v$, it is possible to construct the offset rectangular lattice and use it to discretize the targeted surface. 

\begin{figure}[!b]
\centering
\includegraphics[width=0.9\linewidth]{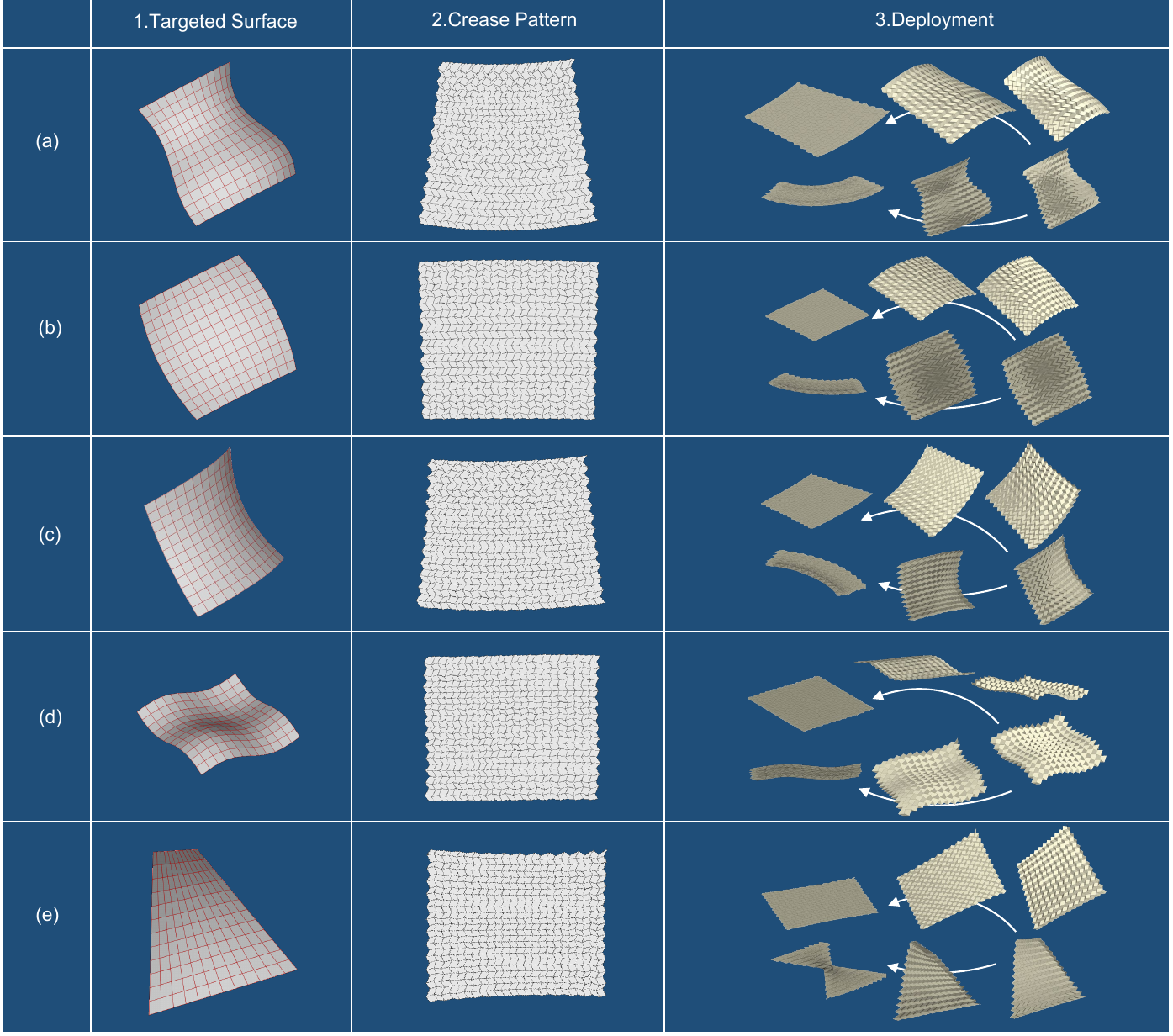}
\end{figure}
\begin{figure}[!t]
\centering
\includegraphics[width=0.9\linewidth]{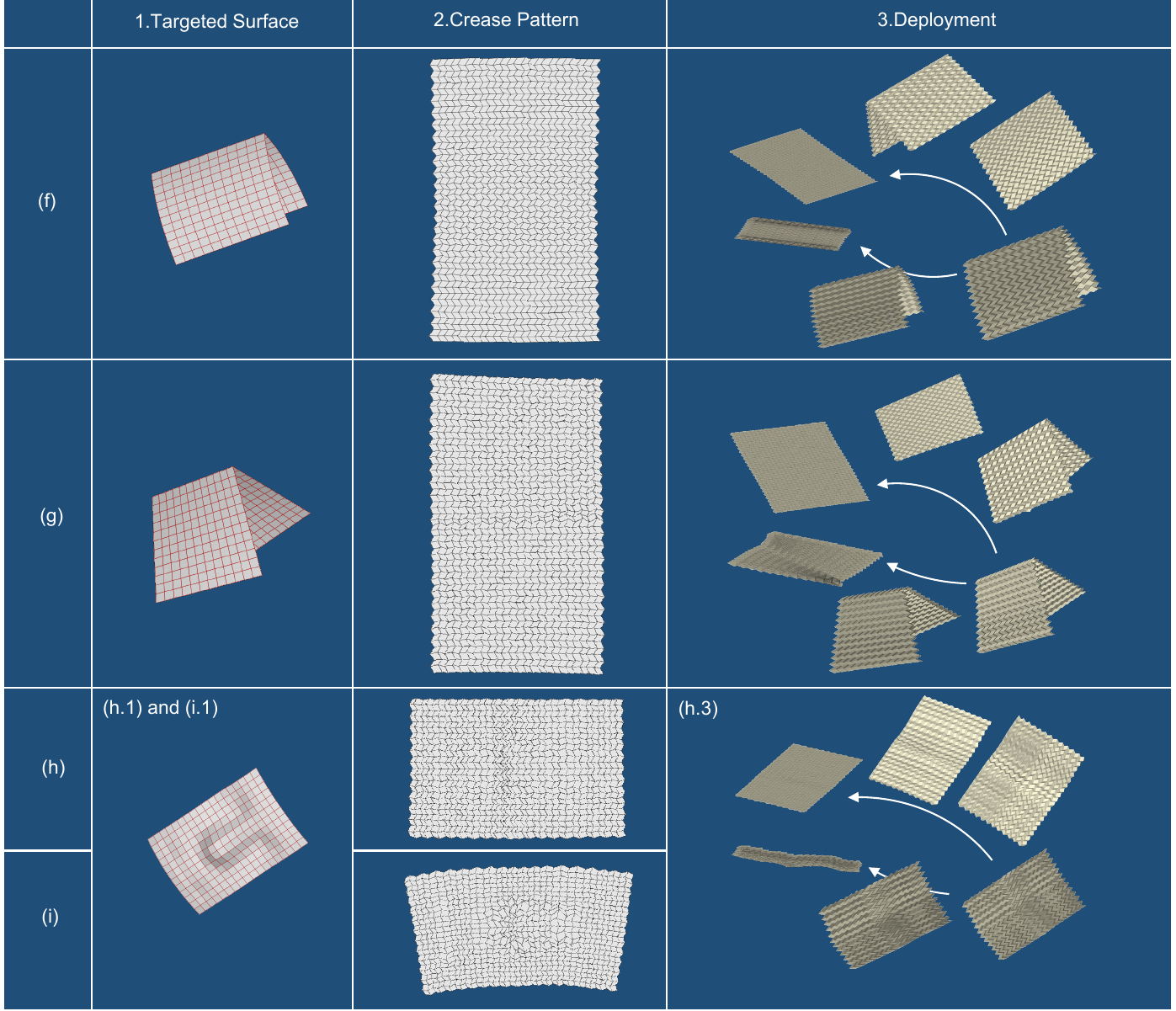}
\caption{The targeted surfaces, reference crease patterns, and deployment processes for the numerical cases: (a) A quarter vase, (b) spherical cap, (c) hyperboloid, (d) 2D sinusoid, (e) saddle, (f) connecting cylinders, (g) connecting saddles, and (h-i) the human-face surface. For cases (h) and (i), we refer to the deployable (h) and non-deployable (i) origami in Fig.~\ref{fig:face}.}
\label{fig:ap-examples}
\end{figure}

\subsection{Change of folding topology} 
\label{sect:ap-topology}
Changing folding topologies from $\sigma = +$ to $\sigma = -$ on the boundary will lead to unusual folding behaviors by the marching algorithm. In fact, sharp ridges of origami will emerge, as depicted in Figs.~\ref{fig:ap-deformation}{(a.3-a.6)}. By observation, the change of folding topology on the left or bottom boundary has different effects on the manners of folding. Specifically, by changing $\sigma=+$ to $\sigma = -$ on the bottom boundary in Figs.~\ref{fig:ap-deformation}{(a.3) and (a.4)}, we observe that a ``step" in $i$ direction emerges in the folded state. Differently, by changing $\sigma=+$ to $\sigma = -$ on the left boundary in Figs.~\ref{fig:ap-deformation}{(a.5) and (a.6)}, a "V" shape changing in $j$ direction  is observed in the folded state. These two observations inspire us to approximate target surfaces with sharp interfaces. For example, our results in Figs.~\ref{fig:examples}{(f) and (g)} exploit the change of folding topology on the left boundary and result in good approximations.

\section{Approximation results}
In Fig.~\ref{fig:ap-examples}, we provide the optimal reference and deformed crease patterns for all the approximation cases, as well as their deployments.    

\section{Computational resources}
\label{sect:ap-resources}
We use the function fmincon in Matlab (R2019b) Optimization Toolbox to perform the sequence quadratic program (SQP) algorithm in both Shape and Point Optimization.
{ The gradients and Hessians used in the SQP solver are computed by the forward difference scheme.}
We use the function pcregistercpd in Matlab (R2019b) Computer Vision Toolbox to perform the coherent point drift (CPD) algorithm in the registration step. The RFFQM cases we provide in { Figs.~\ref{fig:examples} (a-g)} and Fig.~\ref{fig:face}(b) are computed on a laptop with the Intel(R) Core(TM) i7-9750H CPU (single-thread serial). The human-face case in Fig.~\ref{fig:face}(c) is computed on the High-performance Computing (HPC) Platform of Peking University (32-thread parallel).  

\end{appendix}

\end{document}